\documentclass[11pt]{article}
\usepackage{fullpage, amsmath,amssymb, amsthm, mathrsfs,verbatim}
\usepackage{amssymb,multicol}
\usepackage{xspace}
\usepackage{xcolor}
\usepackage{subfig}
\usepackage{hyperref}
\usepackage{tabularx}
\usepackage{url}
\usepackage{algorithm}
\usepackage[noend]{algpseudocode}
\usepackage{enumerate}

\definecolor{verde}{rgb}{0.,0.6,0.2}
\definecolor{bianco}{rgb}{1.,1.,1.}
\definecolor{marrone}{rgb}{0.7,0.2,0.1}
\definecolor{rosso}{rgb}{1,0,0}
\definecolor{giallo}{rgb}{1.0, 0.87, 0.0}
\definecolor{blu}{rgb}{0.03, 0.27, 0.49}
\definecolor{daffodil}{rgb}{0.03, 0.27, 0.49}
\definecolor{darkcerulean}{rgb}{1.0, 0.87, 0.0}

\newcommand{\blue}[1]{{\color{daffodil}{#1}}}

\newcommand{\white}[1]{{\color{bianco}{#1}}}

\newcommand{\red}[1]{{\color{rosso}{#1}}}

\newcommand{\calP}{\ensuremath{\mathbb{P}}\xspace}

\newcommand{\calT}{\ensuremath{\mathcal{T}}\xspace}
\newcommand{\calH}{\ensuremath{\mathcal{H}}\xspace}
\newcommand{\pp}{\ensuremath{\mathbb {P}}\xspace}

\DeclareMathOperator{\CSP}{\mathsf{CSP}}
\DeclareMathOperator{\CAP}{\mathsf{CAP}}
\DeclareMathOperator{\DLP}{\mathsf{DLP}}

\usepackage{hyperref,graphicx}

\usepackage{tikz}
\usepackage{qtree}
\usepackage{multirow}

\usetikzlibrary{matrix}
\usetikzlibrary{positioning}
\usetikzlibrary{shapes.geometric}
\usetikzlibrary{shapes,decorations,shadows}  
\usetikzlibrary{decorations.pathmorphing}  
\usetikzlibrary{decorations.shapes}  
\usetikzlibrary{fadings}  
\usetikzlibrary{patterns}  
\usetikzlibrary{calc}  
\usetikzlibrary{automata}
\usetikzlibrary{chains,fit,shapes}
\usetikzlibrary{backgrounds}

\newtheorem{theorem}{Theorem}[section]
\newtheorem{definition}[theorem]{Definition}
\newtheorem{assumption}[theorem]{Assumption}
\newtheorem{claim}[theorem]{Claim}
\newtheorem{lemma}[theorem]{Lemma}
\newtheorem{corollary}[theorem]{Corollary}

\usepackage{titlesec}

\usepackage{graphicx}
\usepackage{mwe}
\usepackage{hyperref}
\hypersetup{pdftoolbar=true,pdfpagemode=UseNone,pdfstartview=FitH,
  colorlinks=true,extension=}

\newcounter{multifig}

\newcommand{\figcaption}[2]
{\stepcounter{multifig}
\addcontentsline{lof}{figure}{\string\numberline {\arabic{multifig}}{\ignorespaces #1}}
Fig. \arabic{multifig}: #1 #2}

\title{Tight bounds for maximal  identifiability of failure nodes in Boolean network tomography\thanks{A preliminary version of this paper appeared in \cite{DBLP:conf/icdcs/GalesiR18}. }}

\author{Nicola~Galesi \\ Department of Computer Science \\ Sapienza Universit\`a di Roma, Italy \\ {\tt nicola.galesi@uniroma1.it} \and 
Fariba~Ranjbar \\ Department of Computer Science \\ Sapienza Universit\`a di Roma, Italy \\
{\tt fariba.ranjbar@uniroma1.it}}

\begin{document}

\maketitle

\begin{abstract}
We study maximal identifiability, a measure recently introduced in  Boolean Network Tomography to characterize networks' capability to localize failure nodes in end-to-end path measurements. We prove  tight upper and lower bounds on the maximal identifiability of failure nodes for specific classes of network topologies, such as trees and  $d$-dimensional grids, in both directed and undirected cases. We prove that directed $d$-dimensional grids with support $n$ have maximal identifiability $d$ using $2d(n-1)+2$ monitors;  and in the undirected case  we show that $2d$ monitors suffice to get identifiability of $d-1$. We then study identifiability under embeddings: we establish relations between maximal identifiability, embeddability and graph dimension when network topologies are modeled as DAGs.  Our results suggest the design of networks over $N$ nodes with maximal identifiability  $\Omega(\log N)$ using  $O(\log N)$ monitors and a heuristic to boost maximal identifiability on a given network by simulating $d$-dimensional grids. We  provide positive evidence of this heuristic  through data extracted by exact computation of maximal identifiability on examples of small real networks. 
\end{abstract}

\section{Introduction} Monitoring a network to localize corrupted components is essential to guarantee a correct behaviour and the 
reliability of a network. In many real networks direct access and direct monitoring of the individual components are not possible 
(for instance because of limited access to the network) or unfeasible in terms of available resources (protocols, communications, 
response-time etc.).   A well-studied approach to localization of failing components  is {\em network tomography}. Network tomography focuses 
on detecting the state of single components in the network by  running  a {\em measurement} process along the network.
The process starts by sending packets (containing suitable data to capture interesting failures)
 from specific {\em source-monitor} nodes   and terminates receiving another data packet 
on other specific {\em target-monitor} nodes.

Measurement is done along a set of {\em end-to-end} 
paths, each one starting and ending with a monitor node. 
In this work we  focus on the problem of detecting {\em node states} (failing/working), using  
a {\em Boolean network tomography} approach \cite{DBLP:conf/imc/Duffield03,DBLP:journals/tit/Duffield06} where the received data at each monitor is one bit  (failure (1) /working (0)), 
capturing the presence or the absence of a failure along a path. 
We are interested in identifying (uniquely) failure nodes. Receiving a $0$ (working state) at an end monitor of a 
path means that each node in the path  is working properly.  Then  the localization of failing nodes in a set of paths \calP (or a network viewed as a set of paths) is captured by the solutions to the following Boolean system: 

\begin{eqnarray}
\label{eqn:1}
\bigwedge_{p \in \calP} \left( \bigvee_{v\in p} x_v\equiv b_p \right )
\end{eqnarray}

where $\vec b$ is a vector of Boolean values (corresponding to final measurement in the paths)  
and $x_v$'s are Boolean variables, one  for each node $v$. 
Any solution to this system is a possible location of node-failures  satisfying the measurements.   

\subsection{The problem and related work}
A set of (non monitor) nodes failing simultaneously is a failure set. Each solution to Equation \ref{eqn:1} 
 captures a failure set that can occur in the network according  to the measurements. But as readily seen
 solutions to Eq. \ref{eqn:1}  are often multiple. In \cite{DBLP:journals/ton/HeGMLST17,DBLP:conf/imc/MaHSTLL14,DBLP:journals/pe/MaHSTL15}, the authors proposed a parameter, later refined in 
 \cite{DBLP:journals/ton/MaHSTL17} measuring the ability of a network of capturing the maximum number of simultaneous failure nodes which are  uniquely identifiable. This measure is called {\em maximal identifiability} (Definition \ref{def:kid}).  
 Maximal identifiability for detection of failed nodes  in the Boolean case was recently studied under several aspects,
 including network topologies, routing protocols and probing mechanism.
 Ma et. al. in \cite{DBLP:journals/ton/MaHSTL17,DBLP:journals/pe/MaHSTL15}  
 investigated network topology questions such as under what conditions one can uniquely localize failed nodes  
from path measurements available in the entire network or  what is the maximum number of simultaneous
node failures that can be uniquely localized.  Answers to such questions depend on network
topology, placement of monitors, and the implemented routing mechanism, as \cite{DBLP:journals/ton/MaHSTL17,DBLP:journals/pe/MaHSTL15} showed.

This works 
are focused on  improving monitoring scheme and heuristics for network design
with the aim of maximizing the number of identifiable nodes in a general network setting.

An important aspects of end-to-end measurements paths is how data are routed through the paths. Practical routing concerns with the probing mechanism: routing protocols and  probing schemes can play a fundamental role in analysing maximal identifiability, since they can restrict the set of paths under consideration. 
In the works  \cite{DBLP:journals/ton/MaHSTL17,DBLP:journals/pe/MaHSTL15} they considered  the 
following classes of probing mechanisms: (1)
\emph{Controllable Arbitrary-path Probing} (CAP), which includes any path/cycle, allowing repeated nodes/links, provided
each path/cycle starts and ends at (the same or different) monitors; (2) \emph{Controllable Simple-path Probing (CSP)} which includes any simple (i.e., cycle-free) path between different monitors; (3) {\em Uncontrollable Probing} (UP): the set of paths between monitors is determined by the routing protocol used by the network, not controllable by the monitors.  Such routing mechanisms find practical implementations as showed in \cite{DBLP:journals/ton/Hu0W0L0ZG16}. In this work  we also
focus on such routing mechanisms. 

With the aim of optimizing the maximal identifiability of a given network  many recent works on node identifiability \cite{DBLP:journals/ton/MaHSTL17,DBLP:journals/pe/MaHSTL15,DBLP:journals/corr/abs-1903-10636}  focus on heuristics/strategies to properly increase the number monitors and to decide where to place them on the internal nodes of the network. However structural limitations due to the network topology might affect the feasibility of such approaches. For instance, as we notice in this work, the minimal degree of the graph modeling the network is a structural limit on the maximal node failure identifiability one can hope for independently of the monitors. 
Despite of the  evident practical implications of failure detection in network reliability and of the recent theoretical and experimental  studies on  maximal identifiability, there is still lack of complete understanding of what maximizing failure node identifiability 
requires in terms of network properties as the topology and  the monitor placement, in particular  if we assume  the more 
general routing protocols. Our work contributes to this line of research.

\subsection{Overview of results}
Through a combinatorial approach we focus on:
\begin{enumerate}
\item studying how structural properties of the graph modeling the network limit maximal node failure identifiability; 
\item studying tight upper and lower bounds for maximal node failure identifiability in specific topologies, in particular for trees and hypergrids, possibly independently of the monitor placement; 
\item understanding how embeddability between graphs interferes with  the maximal identifiability; 
\item exploring experimentally the feasibility of a heuristic to boost maximal identifiability in a network by adding edges in order to approximate  a hypergrid. 
\end{enumerate}
Identifiability as defined for the first time in \cite{DBLP:conf/imc/MaHSTLL14} captures the combinatorial property that to separate two sets $U$ and $W$ (of failure nodes) one wants to exhibit a measurement path in \calP  touching nodes of exactly one of the two sets. The maximal  size of  sets  of failure nodes one can guarantee identifiability for, is then  a measure of the ability to identify failure sets uniquely using paths in \calP.  We want to explore this property, independently of the monitor placement,  for specific classes of topologies such as trees and grids and hypergrids, 
which are among the topologies most used and implemented in real networks. 

To study the maximal identifiability of given graphs $G$, we follow the approach initiated in  \cite{DBLP:journals/ton/MaHSTL17,DBLP:conf/imc/MaHSTLL14,DBLP:journals/pe/MaHSTL15,DBLP:journals/ton/HeGMLST17} based on the definition of maximal identifiability. Given  a graph 
$G=(V,E)$ and  a monitor placement $\chi$ for $G$ we work with the set of paths \calP which are definable 
according to a probing mechanism on $G$ with the monitors assigned by $\chi$. 
We study the maximal identifiability of the set of nodes in $G$ appearing in paths in \calP.

We prove upper and lower bounds on the maximal identifiability of specific classes of network topologies, such as trees, $d$-dimensional hypergrids, in both directed and undirected cases. One first result we obtain is that when the graph $G$ is a tree its maximal node identifiability is very low, namely 1. This result 
has to be interpreted as saying that if our network topology is a tree then maximal number of failed nodes we can hope to uniquely identify is 
1.   Searching for topologies which are better than trees with respect to maximal node-failure identifiability we considered the case of grids.
We prove that grids, under a suitable optimal monitor placement, can reach an identifiability strictly greater than 1, namely 2.
Our analysis easily extends to the case of $d$-dimensional hypergrids where we prove that the maximal identifiability can be lifted to the 
dimension $d$. We prove these results for  both  the directed and undirected cases.

When one consider the minimal number of monitors to reach the maximal identifiability on  $d$-dimensional hypergrids, our results mark an important difference between the directed and undirected cases. In the latter we can show how to get tight lower and upper bound results using only $2d$  monitors.  In the directed case instead  the number of monitors to reach a maximal identifiability 
depends linearly on the number of nodes and cannot be improved. 
 
To prove lower bounds on maximal identifiability, instead of checking experimentally the optimality of the upper bounds as in previous works, we use an algorithmic/combinatorial analysis, so obtaining tight results. 
This approach directly leads to algorithms to design network topologies with a guarantee of reaching a precise  maximal identifiability of failure nodes.

As seen, $d$-dimensional hypergrids play an important role in our results. It is well known that hypergrids are related to the \emph{dimension} 
of directed acyclic graphs (DAG) through the operation of \emph{embeddings} of graphs. Namely the dimension of a DAG $G$, is the smallest
integer $d$ such that $G$ is embeddable in  the $d$-dimensional hypergrid.  We start the study of maximal  identifiability of node failure under embeddings of DAGs. We establish relations between maximal identifiability and embeddability when networks are modeled by DAGs. 
While the most general definition of embeddings can drastically decrease maximal identifiability,  yet we explore two directions: (1) restricting the class of  topologies  we want to embed and (2) restricting the mapping  that defines the embedding. In both cases we show significative results on how maximal identifiability can be preserved under embeddings.

 $d$-hypergrids are examples of concrete topologies  which reach a very good value of  
the maximal identifiability. The results on embeddability and on the dimension suggest that for increasing the maximal identifiability 
of real networks (which often are very low since many real topologies are trees, quasi-trees or grids) one can try to add edges to
the network in such a way to get closer to a graph which is embeddable into a $d$-hypergrid, for $d$ a  function of the 
number of nodes in the network. We explore this idea experimentally.  Namely we propose a simple algorithm {\tt Agrid} that given a network $G$  and a parameter  $d$ outputs a new network $G^{\tt A}$ closer to a $d$-hypergrid than the original network  having minimal degree $d$. We test experimentally our algorithm on real examples of networks, on random graphs, and also on random placement of monitors, obtaining results which are always positive and promising to boost maximal identifiability on real networks. We discuss examples of cost-benefit tradeoffs to evaluate feasibility of {\tt Agrid} on real networks.      


 
 \subsection{Organization}
 In \emph{Section \ref{sec:prel}} we include all the preliminary definitions, including definitions related  with maximal identifiability.

 In \emph{Section \ref{sec:struct-lim}} we start presenting some upper bounds for maximal identifiability in terms of structural properties of the network. We consider: (1) the number of nodes linked to monitors (\emph{Theorem \ref{thm:nummon}}),  (2) the minimal degree  (\emph{Lemma \ref{lem:ubd} and \ref{lem:bd}}), and (3) the number of edges and nodes (\emph{Corollary \ref{cor:nodesandpaths}}).  We discuss consequences of these results.

\emph{Section \ref{sec:dirtop}}  includes the tight bounds on the maximal identifiability of  trees, grids and $d$-dimensional grids in the directed case.  While the upper bounds are derived as consequence of the results on the degree in Section \ref{sec:struct-lim}, in the section we 
prove  the lower bounds (\emph{Theorems \ref{thm:dirtree}, \ref{thm:dirgrid} and \ref{thm:hg}})  for trees, grids and $d$-dimensional grids. We use a specific monitor placement to prove the lower bounds, but we discuss its optimality.
    
 In \emph{Section \ref{sec:und}}  we analyse the undirected case of trees, grids and $d$-hypergrids. As for the previous section the upper bounds are a consequence of the results on the degree in Section \ref{sec:struct-lim} and in the section  we prove the two lower bounds (\emph{Theorems \ref{thm:undtrees} and \ref{thm:undgrid2}}). The $d-1$ lower bound for $d$-hypergids is given for any monitor placement of $2d$ monitors.

\emph{Section \ref{sec:emb}} self-contains all the results about maximal identifiability and isomorphic embeddings. 
In \emph{Theorem \ref{thm:emdcr}} we prove that under a specific routing scheme upper bounds for maximal identifiability (for DAGs) are preserved under any embedding. Later, we restrict the classes of embeddings to those increasing the distance and we can prove that lower bounds on maximal identifiability are preserved under such embeddings (\emph{Theorem \ref{thm:dremb}}). This result in turn  is used to prove that, for DAG closed under transitivity, maximal identifiability  is lower bounded by the dimension of the graph, under isomorphic embeddings (\emph{Theorem \ref{thm:dim}}).
 
In \emph{Section \ref{sec:algo}} we discuss practical applications of our results. First we observe how Theorem \ref{thm:undgrid2} suggests
 the design of a network on $N$ nodes potentially reaching a maximal identifiability of $O(\log N)$. Later in \emph{Subsection \ref{subsecalgo:2}} we describe a heuristic, {\tt Agrid}, to boost maximal identifiability in a network adding random edges in order to increase the original minimal degree. {\tt Agrid} implements the idea of taking a network $G$ and a parameter $d$ and producing a network $G^{\tt A}$ (with the same nodes of $G$ but with minimal degree $d$) and a monitor placement for $G^{\tt A}$ (and $G$) to simulate  a $d$-hypergrid, with the aim of boosting the maximal identifiability of $G$ as close to $d$ as possible. We discuss its feasibility in static, dynamic and sub-networks. 
 
  
\emph{Section \ref{sec:perfo}} is about the report of data on {\tt Agrid} performance.  
We discuss in several Tables and  under four different type of data the  performance of {\tt Agrid} on concrete examples of networks.
 
The last \emph{Section \ref{sec:final}}  contains a discussion on three topics:  degenerate paths,  routing mechanisms and  further research directions.
 
 \section{Preliminaries}
\label{sec:prel}
For sets $U,V$, 
$U\triangle V= (U \setminus V) \cup (V \setminus U)$ is the symmetric difference
between $U$ and $V$. 
In a graph $G=(V,E)$, $V$ is a set of nodes and  $E \subseteq V \times V$. $G$ is undirected if pairs in $E$ are unordered.
 Otherwise $G$ is directed. $G$ is DAG if it is directed and with no cycles. 
A path $p$ in $G$  from  a node $u$ to a node $v$ is a sequence of edges $p=(u_1u_2),(u_2u_3)\dots(u_{k-1}u_k)$ such that
$u_1=u$ and $u_k=v$ and $(u_iu_{i+1})\in E$ for all $i\in[k-1]$. If $G$ is a DAG, then we identify the path $p$ also with sequence of nodes $u_1\ldots u_k$.
For a node $u$ in $G$, $N(u)$ is the set of neighbours of $G$, i.e. $\{v \in V \;|\; (uv)\in E\}$. The degree of $u$,  $\deg(u)$,  is the cardinality of 
$N(u)$. The degree of $G$ is $\Delta(G)=\max _{u\in V}\deg(u)$. We also consider the minimal degree $\delta(G)$ of $G$. 
If $G$ is directed then we distinguish $N_{\tt i}(u)$, the set of neighbours $v$ of $u$ s.t. $(vu)\in E$, from 
$N_{\tt o}(u)$, the neighbours $v$ of $u$ s.t. $(uv)\in E$. For all degree measures on $G$ we distinguish the in-degree $\Delta_{\tt i}(G)$ and $\delta_{\tt i}(G)$ and the out-degree $\Delta_{\tt o}(G)$, and  $\delta_{\tt o}(G)$.  

\begin{table}[h!]
  \begin{center}
   \scalebox{1}{
    \begin{tabular}{l|l} 
      \textbf{Symbol} & \textbf{Meaning} \\
      \hline
      \hline
      $V$ & set of nodes \\
      \hline
      $E$ & set of edges \\
       \hline
      $N(u)$ & neighbours of $u$ \\
        \hline
     $\deg(u)$ & degree of $u$, i.e. $|N(u)|$ \\        
        \hline
        $\Delta(G)$ & maximal degree in $G$\\ 
       \hline
       $\delta(G)$ & minimal degree in $G$\\
        \hline
        $\Delta_{\tt o}(G), \Delta_{\tt i}(G)$ & maximal out (resp. in)-degree in directed $G$\\ 
       \hline
     $\delta_{\tt o}(G), \delta_{\tt i}(G)$ & minimal out (resp. in)-degree in directed $G$\\ 
       \hline
    \end{tabular}
    }
     \caption{\em Graph notations on $G=(V,E)$.}
    \label{fig:graphnot}
  \end{center}
\end{table}

\smallskip
\noindent{\bf \em Topologies}. We consider the following graphs. 
Let  $d\in \mathbb N^+$ and $n \in \mathbb N$, $n \geq 4$. The \emph{(directed) hypergrid} of dimension $d$ over support $[n]$,  $\calH_{n,d}$,  is the graph with  vertex set $[n]^d$ and where there is a directed  edge from a node $x = (x_1, x_2, . . . , x_d)$ to  a node $y = (y_1, y_2, . . . , y_d)$ if for some $i\in [d]$ we have $y_i - x_i = 1$ and $x_j = y_j$ for all $j \not= i$.

In the case of {\em undirected hypergrid}  in $\calH_{n,d}$ there is an edge between  a node  $x$ and a node $y$ if for some 
$i\in [d]$ we have $|x_i - y_i| = 1$ and $x_j = y_j$ for all $j \not= i$.
In the case of simple grids over $n$ nodes, i.e. $d=2$, we  use the notation $\calH_n$. $\partial_i$ is the set of nodes $x = (x_1, x_2, . . . , x_d)$ such that $x_i=1$ .
A {\em border node} is a node of $\calH_{n,d}$ which is also in some $\partial_i$.

\begin{figure}[h!]
\centering
\begin{center}
\begin{scriptsize}
\tikz[scale=1.4] {

\node at (-0.5,0) () {$(1,1)$};
\coordinate (a) at (0,0);
\draw [] (a) circle [radius=0.05];
\coordinate (b) at (0.5,0);
\draw [] (b) circle [radius=0.05];
\coordinate (c) at (1,0);
\draw [] (c) circle [radius=0.05];
\coordinate (d) at (1.5,0);
\draw [] (d) circle [radius=0.05];

\draw[draw=black, ->]  (0+0.05,0) -- (0.5-0.05,0);
\draw[draw=black, ->]  (0.5+0.05,0)--(1-0.05,0);
\draw[draw=black, ->]  (1+0.05,0)--(1.5-0.05,0);

\coordinate (a1) at (0,-0.5);
\draw [] (a1) circle [radius=0.05];
\coordinate (b1) at (0.5,-0.5);
\draw [] (b1) circle [radius=0.05];
\coordinate (c1) at (1,-0.5);
\draw [] (c1) circle [radius=0.05];
\coordinate (d1) at (1.5,-0.5);
\node at (2,0) () {$(4,1)$};
\node at (2,-1.5) () {$(4,4)$};
\node at (-0.5,-1.5) () {$(1,4)$};
\draw [] (d1) circle [radius=0.05];
\draw[draw=black, ->]  (0+0.05,-0.5) -- (0.5-0.05,-0.5);
\draw[draw=black, ->]  (0.5+0.05,-0.5)--(1-0.05,-0.5);
\draw[draw=black, ->]  (1+0.05,-0.5)--(1.5-0.05,-0.5);
\draw[draw=black, ->]  (0,0-0.05) -- (0,-0.5+0.05);
\draw[draw=black, ->]  (0.5,0-0.05) -- (0.5,-0.5+0.05);
\draw[draw=black, ->]  (1,0-0.05) -- (1,-0.5+0.05);
\draw[draw=black, ->]  (1.5,0-0.05) -- (1.5,-0.5+0.05);

\coordinate (a2) at (0,-1);
\draw [] (a2) circle [radius=0.05];
\coordinate (b2) at (0.5,-1);
\draw [] (b2) circle [radius=0.05];
\coordinate (c2) at (1,-1);
\draw [] (c2) circle [radius=0.05];
\coordinate (d2) at (1.5,-1);
\draw [] (d2) circle [radius=0.05];
\draw[draw=black, ->]  (0+0.05,-1) -- (0.5-0.05,-1);
\draw[draw=black, ->]  (0.5+0.05,-1)--(1-0.05,-1);
\draw[draw=black, ->]  (1+0.05,-1)--(1.5-0.05,-1);
\draw[draw=black, ->] (0,-0.5-0.05) -- (0,-1+0.05);
\draw[draw=black, ->]  (0.5,-0.5-0.05) -- (0.5,-1+0.05);
\draw[draw=black, ->]  (1,-0.5-0.05) -- (1,-1+0.05);
\draw[draw=black, ->]  (1.5,-0.5-0.05) -- (1.5,-1+0.05);

\coordinate (a3) at (0,-1.5);
\draw [] (a3) circle [radius=0.05];
\coordinate (b3) at (0.5,-1.5);
\draw [] (b3) circle [radius=0.05];
\coordinate (c3) at (1,-1.5);
\draw [] (c3) circle [radius=0.05];
\coordinate (d3) at (1.5,-1.5);
\draw [] (d3) circle [radius=0.05];
\draw[draw=black, ->]  (0+0.05,-1.5) -- (0.5-0.05,-1.5);
\draw[draw=black, ->]  (0.5+0.05,-1.5)--(1-0.05,-1.5);
\draw[draw=black, ->]  (1+0.05,-1.5)--(1.5-0.05,-1.5);
\draw[draw=black, ->] (0,-1-0.05) -- (0,-1.5+0.05);
\draw[draw=black, ->]  (0.5,-1-0.05) -- (0.5,-1.5+0.05);
\draw[draw=black, ->]  (1,-1-0.05) -- (1,-1.5+0.05);
\draw[draw=black, ->]  (1.5,-1-0.05) -- (1.5,-1.5+0.05);


}
\end{scriptsize}
\end{center}
\caption{\em Directed hypergrid $\calH_4=\calH_{4,2}$.}
\label{fig:H4}
\end{figure}

We  consider (see Figure \ref{fig:dirtrees}) directed rooted trees  $\calT_n$ over $n \in \mathbb N^+$ nodes (from now on we omit that $n\in \mathbb N^+$). An (undirected) tree is an acyclic graph with no cycle where any two nodes $u$ and $v$ are  connected by a path. We consider: (1) {\em downward directed trees} $\calT_n$ where the root of $\calT_n$ is the only source node and the leaves of $\calT_n$ are the only target nodes (i.e. $\Delta_{\tt i}(\calT_n)\leq 1$)  (2) {\em upward directed trees} $\calT_n$, where the root is the only target node and the leaves are source nodes (i.e. $\Delta_{\tt o}(\calT_n)\leq 1$).
\begin{figure}[H]
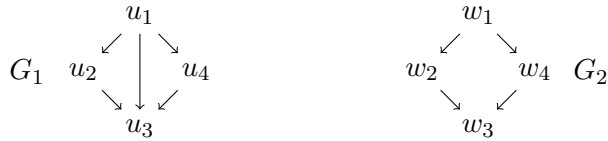

\centering
\begin{center}
\tikz[scale=1.5]{
\node (g) at  (-1,-.5) {$G_1$};
\node (u1) at  (0,0) {$u_1$};
\node (w1) at  (-.5,-.5)  {$u_2$};
\node (u2) at  (0,-1) {$u_3$};
\node (w2) at  (0.5,-.5) {$u_4$};

\draw[->] (u1)--(u2);
\draw[->] (u1)--(w1);
\draw[->] (u1)--(w2);
\draw[->] (w1)--(u2);
\draw[->] (w2)--(u2);

\node (g2) at  (4,-.5) {$G_2$};
\node (au1) at  (3,0) {$w_1$};
\node (w11) at  (2.5,-.5)  {$w_2$};
\node (au2) at  (3,-1) {$w_3$};
\node (aw2) at  (3.5,-.5) {$w_4$};

\draw[->] (au1)--(w11);
\draw[->] (au1)--(aw2);
\draw[->] (w11)--(au2);
\draw[->] (aw2)--(au2);
}
\end{center}
\caption{\em Example of embedding $G_1 \hookrightarrow_f G_2$:  $f(u_i)=w_i$.}
\label{fig:embedding}
\end{figure}

\noindent{\bf \em Embeddings}.
Each DAG $G = (V, E)$ is equivalent to a {\em poset} with elements $V$  and partial order $\preceq_G$, where $u \preceq_G v$ if $v$ is
reachable from $u$ in $G$. Elements $u$ and $v$  are {\em comparable} if $u\preceq_G v$  or $v \preceq_Gu$ , and {\em incomparable} otherwise.
We write $u \prec_G v$  if $u \preceq_G v$ and $u \not = v$. 
A mapping $f$ from a poset $G=(V,E)$ to a poset $G'=(V',E')$
is called an {\em embedding} if $f$ is injective and it respects the partial order, that is, all
$u,v \in V$ are mapped to $u',v'\in V'$ such that $u \preceq_G v$ iff $u' \preceq_{G'}v'$. If $G$ is embeddable into $G'$ we write 
$G \hookrightarrow G'$. 
$w_1=f(u_1)\preceq f(u_3)=w_3$ since in $G_2$ there is a path from $w_1$ to $w_3$.

\noindent{\bf \em Paths, monitors and identifiability}.
\label{subsec:assume}
Let \calP be a set of paths over nodes $V$. For a node $v\in V$, let  $\calP(v)$ be the set of paths in \calP passing through $v$.  
For a set of nodes $U$, $\pp(U) =\bigcup _{u\in U}\pp(u)$. Hence if $U\subseteq V$, $\pp(U)\subseteq \pp(V)$.
In end-to-end measurement paths, messages are routed  and received through monitor nodes.  We  work with  the assumption that: \emph{physical monitors are external to the network}.
This is justified by two reasons: (1) Monitors by default must be reliable, hence there is no failure to identify for them; 
(2) Since we study maximal identifiability in set of paths associated to given topologies $G=(V,E)$, the assumption allows to 
consider all the nodes in $G$ as equally potentially identifiable for a failure.

Let $I,O$ be sets of physical monitors. A \emph{monitor placement} for $G=(V,E)$ is a pair of injective mappings 
$\chi=(\chi_{\tt i},\chi_{\tt o})$ such that $\chi_{\tt i}:I \rightarrow V$ and $\chi_{\tt o}:O \rightarrow V$. 
We always denote by $(\mathfrak{m},\mathfrak{M})$ the pair $(\chi_{\tt i}(I),\chi_{\tt o}(O))$, where clearly $\mathfrak m=\bigcup_{i\in I} \chi_{\tt i}(i)$ and $\mathfrak M=\bigcup_{i\in O} \chi_{\tt o}(i)$.  The interpretation is that $\mathfrak m$ is the set of the nodes in $G$ (\emph{input nodes}) linked to input monitors and $\mathfrak M$  (\emph{output nodes}) the nodes in $G$ linked to output monitors. We use to denote measurement paths in $G$ under $\chi$ as  $\mathfrak m\!\cdot \!(v_1v_2) \cdots (v_{k-1}v_k) \!\cdot \! \mathfrak M$ where $v_1 \in \mathfrak m$ is an input node, $v_k \in \mathfrak M$ is an output node and $(v_1v_2) \cdots (v_{k-1}v_k) $ a path in $G$.
Given a  graph $G=(V,E)$ and a monitor placement $\chi=(\mathfrak m,\mathfrak M)$ we denote by $\calP(G|\chi) $ the set of all \emph{distinct} paths
from a node in $\mathfrak m$ to a node in $\mathfrak M$. 
Let \calP be a set of paths over a set of nodes $N$. Following  \cite{DBLP:journals/ton/MaHSTL17} we define:
\begin{definition}[$k$-identifiability]
\label{def:kid}
$N$ is $k$-identifiable with respect to \calP if and only if 
for all $U,W \subseteq N$, with $U\triangle W \not =\emptyset$ and  $|U|,|W|\leq k$, it holds that $\pp (U)\triangle \mathbb \pp (W) \not =\emptyset$.
\end{definition}

In \cite{DBLP:journals/ton/MaHSTL17}, and later in  \cite{DBLP:journals/corr/abs-1903-10636}, $k$-identifiability was used to localize failure
within specific subsets $S$ of $V$. That definition is given restricting  the condition  $U\triangle W \not =\emptyset$ 
to $(U\cap S) \triangle (W \cap S) \not =\emptyset$. When we need to distinguish our measure from the original one in 
\cite{DBLP:journals/ton/MaHSTL17}, we call the latter \emph{local} identifiability. 

\begin{definition}[Maximal identifiability]
The maximal identifiability of $V$, $\mu(V)$, with respect to 
$\calP$ is the $\max_{k\geq 0}$ such that $V$ is  $k$-identifiable with respect to \calP. 
\end{definition}

Monotonicity of identifiability (a property noticed in  several works \cite{DBLP:conf/imc/MaHSTLL14,DBLP:journals/ton/MaHSTL17}), i.e. that $k$-identifiability implies $k'$-identifiability for $k'<k$, is trivial  from our definition. 

\begin{table}[H]
  \begin{center}
    \scalebox{1}{
    \begin{tabular}{l|l} 
      \textbf{Symbol} & \textbf{Meaning} \\
      \hline
      \hline
      $\calP$ & set of paths \\
       \hline
      $\pp(u)$ & paths in $\calP$ passing through  $u$ \\
        \hline
     $\pp(U)$ & $\bigcup_{u\in U} \pp(u)$ \\        
       \hline
      $(I,O)$  &  physical input and output monitors\\ 
        \hline
      $\chi$  & monitor placement: $\chi=(\chi_{\tt i}(I),\chi_{\tt o}(O))$\\ 
       \hline
       $(\mathfrak m,\mathfrak M)$  & nodes in $V$ linked to $I$ and $O$ by $\chi$ \\ 
        \hline
      $\calP(G|\chi)$ & set of all paths in $G$ from $\mathfrak m$ to $\mathfrak M$\\ 
       \hline
     $\mu(G| \chi)$ & maximal identifiability of $V$ wrt $\pp(G|\chi)$\\
       \hline
    \end{tabular}}
   \caption{\em Notations for paths, monitors, identifiability.}
   \label{fig:pathsnot}
  \end{center}  
\end{table}

\noindent{\bf \em Degenerate paths and maximal identifiability}.
In the Boolean system as in Equation \ref{eqn:1}
we can have equations made by only one variable, $x_v=b$ for some $b\in\{0,1\}$.  This situation might occur when a node $v$ is linked to both input and output monitors. Nevertheless one node alone does not correspond to any real path in a graph and this creates an asymmetry between the system and the set of paths. Since we analyze maximal identifiability from the point of view of (real) paths in a graph, we {\em force} such equation to corresponds to a loop path of one node $\mathfrak m\!\cdot\! (vv) \!\cdot \!\mathfrak M$.  We call this a $\DLP$(-node or -path) from {\em degenerate loop path}. 
While $\DLP$ where previously appeared in the literature on node failure identification, forcing them into a node-loop is new to our knowledge. 

In this work  we consider routing mechanism where $\DLP$ paths are not allowed\footnote{Notice that this only make more difficult to prove lower bounds for maximal identifiability (see Section \ref{subsubsec:lb}).}. In the final section (Sec.\ref{sec:final}) we argue precisely why this assumption is theoretically sound and practically feasible when we consider node failure localization.   

\smallskip

\noindent{\bf \em Routing mechanisms and set of paths}. 
Given the topology $G$ and the monitor placement $\chi=(\mathfrak m,\mathfrak M)$, the probing mechanism plays a crucial
role in determining the set of measurement paths.
We consider three probing mechanisms (see also \cite{DBLP:journals/ton/MaHSTL17}):
\begin{enumerate} 
\item \emph{Controllable Arbitrary-path Probing} ($\CAP$) which includes in $\pp(G|\chi)$ any path/cycle, allowing repeated nodes/links,
provided each path/cycle starts and ends at (the same or different) input/output nodes.
\item \emph{Controllable Arbitrary-path Probing with no $\DLP$} ($\CAP^-$). This is as $\CAP$ but not allowing degenerate loop-paths.
\item  \emph{Controllable Simple-path Probing ($\CSP$)} which allows in $\pp(G|\chi)$ any simple (i.e. cycle-free)
path between different input/output nodes. 

\end{enumerate}

As described in \cite{DBLP:journals/ton/MaHSTL17} and \cite{DBLP:conf/infocom/RenD16}, these probing mechanisms capture the main features of several existing and emerging routing techniques. In Section \ref{sec:final} we discuss in more details these routing mechanisms and their effective implementation.

For a graph $G=(V,E)$ and a monitor placement $\chi$ for $G$, 
we write $\mu(G|\chi)$ (and call it the maximal identifiability of $G|\chi$), to indicate the maximal identifiability of $V$  with respect to $\pp(G|\chi)$ under the routing mechanism considered.
We might  omit the $\chi$,  when it is either clear from the context, or when the result holds for all possible $\chi$.
Notice that  since $\CSP$ does not allow loops, then \emph{$\DLP$ paths are not allowed under $\CSP$}. Our results in the next sections hold for $\CSP$ and $\CAP^-$ routing mechanisms (unless explicitly specified).

\subsubsection{How to prove upper bounds for $\mu$}
\label{subsubsec:ub}
To prove that $\mu(G|\chi)\leq k-1$ it is sufficient to show that $G|\chi$ is  not $k$-identifiable. By Definition \ref{def:kid} this means to show the existence of two distinct node sets $U$ and $W$ of cardinality at most $k$ such that $\pp(U) \triangle \pp(W) = \emptyset$.   
Hence by the monotonicity property of identifiability,  this implies that $\mu(G|\chi) \leq k-1$.  

\subsubsection{How to prove lower bounds for $\mu$}
\label{subsubsec:lb}
If we want to prove that $\mu(G|\chi)\geq k$ for some $k$, then  by Definition \ref{def:kid} it is enough to argue that for all distinct node sets $U$ and $W$ of cardinality $|U|,|W| \le k$,  $\pp(U) \triangle \pp(W) \not = \emptyset$. To prove this, we have to show that for any two distinct node sets $U$ and $W$ of cardinality at most $k$ there exists always a path in \calP intersecting exactly one node set between $U$
and $W$. Lower bounds on $\mu(G|\chi)$ are  hence interesting since to prove them we have to show the existence of paths in \calP  distinguishing between any two node sets $U$ and $W$ of cardinality at most $k$, i.e. touching exactly one of them.

\section{Structural upper bounds on maximal identifiability}
\label{sec:struct-lim}
In this section we  show some upper bounds on maximal identifiability due to structural property of the topology. We consider mainly two aspects: the number of nodes linked to monitors  and the minimal degree of the network.  

 \subsection{Number of input and output nodes}
 \label{subsec:2monitor}

Having monitors external to the network, we look at  the maximal identifiability we can hope for in a graph, knowing how many internal nodes are linked to monitors.
Next theorem answers  such question for the  $\CSP$ routing scheme.

Let $G=(V,E)$ be a graph. Let  $\chi=(\chi_{\tt i},\chi_{\tt o})$ be a monitor placement for $G$ of physical monitors $I$ and $O$.
Let $\mathfrak m=\bigcup_{i\in I} \chi_{\tt i}(i)$ and $\hat m = |\mathfrak m|$. Let $\mathfrak M=\bigcup_{i\in O} \chi_{\tt o}(i)$ and $\hat M =|\mathfrak M|$. 
\begin{theorem}
\label{thm:nummon}
Let $G=(V,E)$ be connected and $\chi$ a monitor placement for $G$. Then, under $\CSP$ routing,  $\mu(G|\chi)\!<\!\max(\hat m, \hat M)$.
\end{theorem}
\begin{proof}
Define $U=\mathfrak m$  and $W=\mathfrak M$. Hence $|U|,|W|\leq \max(\hat m,\hat M)$.  
If $U \not =W$, then since $G$ is connected, there is no way of separating $U$ from $W$ with a path going from an input-node to an output-node. We will always touch both. Then $\pp(U)\triangle \pp(W)=\emptyset$ and hence $\mu(G|\chi)<\max(\hat m,\hat M)$.
If $U=W$ then it must be that  $|U|=|W| \ge 2$, since otherwise  $U=W=\{u\}$ and we would have loop path which is not allowed under $\CSP$.
 Then define $U'= U- \{u\}$ where $u$ is one of the nodes that is both the termination of a path and the source of another ($u$ necessarily exists since $U=W$). It is obvious that $\pp(U')\subseteq \pp(W)$. Now, if $p\in \pp(W)$ is not touching $u$, then $p\in \pp(U')$ since $U'\subset W$. If $p\in \pp(W)$ is touching $u$, then the source of this path is in $U'$ and $p$ touches $U'$ as well, unless the source of $p$ is $u$. If the source of $p$ is $u$, then the termination of this path is in $W-\{u\}=U'$ (since loop path is not allowed under $\CSP$) and touches $U'$ as well. Therefore $\pp(U')\triangle \pp(W)=\emptyset$ and hence $\mu(G|\chi)<\max(\hat m,\hat M)$.
\end{proof}



\subsection{Degree}
\label{subsec:deg}
Next  results hold for any monitor placement in $\CSP$ or $\CAP^-$ and we omit $\chi$. We start with the undirected case.
\begin{lemma}
\label{lem:ubd}
Let $G=(V,E)$ be undirected. Then  $\mu(G) \leq \delta(G)$.
\end{lemma}
\begin{proof} Let $u\in V$ be such that $\deg(u)=\delta(G)$. Fix $U=N(u)$ and $W=\{u\} \cup N(u)$. 
Each path touching $u$ is passing through at least a node in $N(u)$. Hence $\pp(\{u\})\subseteq \pp(N(u))$.
Hence $\pp(W)= \pp(\{u\})\cup \pp(N(u)) = \pp(N(u))=\pp(U)$ and then  $\pp(U)\triangle \pp(W)  = \emptyset$.
We have found two sets $U,W$ of cardinality at most $\delta(G)+1$. Hence $\mu(G) \leq \delta(G)$. 
\end{proof}

Notice that if a node $v$ in $V$ is disconnected in $G$, then $\mu(G)=\delta(G)=0$. 
Hence in the rest of the paper, we assume the graphs always to be connected.  
\begin{corollary}
\label{cor:nodesandpaths}
 Let  $G=(V,E)$ be defined  over $n$ nodes and $m$ edges. Then $\mu(G)\leq \min\{n,\lceil \frac{2m}{n} \rceil\}$.
\end{corollary}
\begin{proof}
Assume a graph $G$ has $n$ nodes and minimal degree $d$. Then there are at least $nd/2$ edges in $G$. 
So $m\geq nd/2$. Hence $d\leq 2m/n$. By Lemma \ref{lem:ubd} $\mu(G)\leq d=2m/n$.
\end{proof}

Let us now consider the directed case. 
Let $G=(V,E)$ be a directed graph and $\chi=(\mathfrak m,\mathfrak M)$ be a monitor placement. 
A node $v\in V$ is called a \emph{complex source} if $v \in m$ and $\deg_{\tt i}(v)>0$ and a \emph{simple source}
if $v \in m$ and $\deg_{\tt i}(v)=0$. Let $K$ (resp. $L$) be the set of complex (resp. simple) source nodes and $R = V\setminus (K\cup L)$. 
We let  $\hat \delta(G) = \min\{\;\min_{v \in R} \deg_{\tt i}(v), \min_{v \in K}(\deg_{\tt i}(v)+\deg_{\tt o}(v))\;\}.$

\begin{figure}[h!]
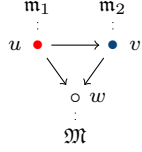

\centering
\begin{center}
\begin{scriptsize}
\tikz {
\node (m2) at  (0.5,0)  {$\mathfrak m_2$};
\node (m1) at  (-.5,0)  {$\mathfrak m_1$};
\node (ul) at  (-.8,-.5)  {$u$};
\node (u) at  (-.5,-.5)  {\red{$\bullet$}};
\node (w) at  (0,-1.2) {$\circ$};
\node (wl) at  (0.3,-1.2) {$w$};
\node (vl) at  (0.8,-.5)  {$v$};
\node (v) at  (0.5,-.5) {\blue{$\bullet$}};
\node (M) at  (0,-1.7) {$\mathfrak M$};
\draw[->] (u)--(w);
\draw[->] (u)--(v);
\draw[->] (v)--(w);
\draw[-,dotted] (m1)--(u);
\draw[-,dotted] (m2)--(v);
\draw[-,dotted] (M)--(w);



}
\end{scriptsize}
\end{center}
\caption{\emph{Example of simple ($\red{\bullet}$) and complex (\blue{$\bullet$}) source nodes.}}
\label{fig:complexsource}
\end{figure}


\begin{lemma}
\label{lem:bd}
Let $G=(V,E)$ be directed. Then $\mu(G) \leq \hat \delta(G)$.
\end{lemma}
\begin{proof}
Let $w$ be a node in $G$ which minimizes $\hat \delta(G)$. If $w \in R$,  then  $\hat \delta(G)=\delta_{\tt i}(w)$. Define $W=N_{\tt i}(w)$ and 
$U=N_{\tt i}(w)\cup \{w\}$. Since $w \in R$, then each path passing through $w$ is necessarily proceeding  from a node in $N_{\tt i}(w)$, hence 
$\pp(\{w\}) \subseteq \pp(N_{\tt i}(w))$. Therefore $\pp(U)=\pp(W)$, which proves the claim since $|U|=\delta_{\tt i}(G)+1=\hat \delta(G)+1$.

If $w \in K$, then define $W=N_{\tt i}(w) \cup N_{\tt o}(w)$ and $U= W \cup \{w\}$. 
If a path is passing from $w$ and raising from an input monitor linked to $w$, it is necessarily  continuing to a node in
$N_{\tt o}(w)$.
\end{proof}

\subsection{Graphs including lines}
We call a path $p$ in an undirected graph  $G=(V,E) $ a  {\em line} if $p:=(u_0u_1)\ldots(u_k u_{k+1})$ and $N(u_i)=\{u_{i-1},u_{i+1}\}$ for any $i\in [k]$ (see also \cite{DBLP:journals/ton/DuffieldP04,DBLP:conf/sigmetrics/BuDPT02}).
Reasoning exactly as in Lemma \ref{lem:ubd} is easy to observe that if $\pp(G|\chi)$ includes a path which is a line, the maximal identifiability of $G$ is less than $1$. 
Hence meaningful topologies should not include a line.  
We define an undirected topology $G$ to be  {\em Line-Free} (LF) if each node $u$ is linked to at least two other nodes in $G$.

 \section{Directed trees and grids}
\label{sec:dirtop}
We now consider directed trees $\calT_n$. For downward trees we consider the monitor placement $\chi_{\tt t}$ which includes in $\mathfrak m$ 
the root of $\calT_n$  and in $\mathfrak M$ all the leaves. Vice versa in an upward tree $\chi_{\tt t}$ assigns the root of $\calT_n$ in $\mathfrak M$ and the leaves
in $\mathfrak m$ (see Figure \ref{fig:dirtrees}).

\begin{theorem} 
\label{thm:dirtree}
Let $\calT_n$ be a directed tree. Then  $\mu(\calT_n|\chi_{\tt t})=1$ under $\CSP$ or $\CAP^-$.
\end{theorem} 
\begin{proof}
We assume the tree to be line-free (LF) so that the bound depends only on the topology and not on the fact that contains a line. 
Consider a node $u$ in $\calT_n$. Since $\calT_n$ is LF $u$ has either in-degree $\geq 2$ or out-degree $\geq 2$. 
According to whether the tree is downward or upward, one of the two cases in Figure \ref{fig:dirtrees} can happen:
\begin{figure}[h!]
\begin{center}
\begin{scriptsize}
\begin{tikzpicture} [scale=1,triangle/.style = {draw, regular polygon, regular polygon sides=3 },
    border rotated/.style = {shape border rotate=180}]
\node[triangle] at (0,0.9) {$\white{t}$};
\node (rm) at (0,2) {$\mathfrak m$};

\node (rm1)  at (0,1.3){};
\draw[->,dotted] (rm) -- (rm1);
\node at (0,0.7) (c) {};
\node at (0,1) {};
\node at (0,0) (d)  {$u$};
\node at (0,-0.1) (f1) {};
\node at (0.5,-0.7) (f3) {$w$};
\node[triangle] at (0.5,-1.3) {$\white{t}$};
\node at (-0.5,-0.7) (f4) {$z$};
\node[triangle] at (-0.5,-1.3) {$\white{t}$};
\draw[->] (f1) -- (f3);
\draw[->] (f1) -- (f4);
\draw[->] (c)--(d);

\node (M1) at (-0.9,-2.3) {$\mathfrak M_1$};
\draw[->,dotted] (-0.9,-1.6)-- (M1);
\node (M11) at (0.5+0.4,-2.3) {$\mathfrak M_{t}$};
\draw[->,dotted] (0.5+0.4,-1.6)-- (M11);
\draw [-,dotted] (M1)--(M11);

\node[triangle] at (4,0.9) {$\white{t}$};
\node at (4,0.7) (1c) {};
\node at (4,1) {};
\node at (4,0) (1d)  {$u$};
\node at (4,-0.1) (1f1) {};
\node at (4.5,-0.7) (1f3) {$w$};
\node[triangle] at (4.5,-1.3) {$\white{t}$};
\node at (3.5,-0.7) (1f4) {$z$};
\node[triangle] at (3.5,-1.3) {$\white{t}$};
\draw[<-] (1f1) -- (1f3);
\draw[<-] (1f1) -- (1f4);
\draw[<-] (1c)--(1d);
\node (m1) at (4-0.9,-2.3) {$\mathfrak m_1$};
\draw[<-,dotted] (4-0.9,-1.6)-- (m1);
\node (m11) at (4+0.5+0.4,-2.3) {$\mathfrak m_{t}$};
\draw[<-,dotted] (4+0.5+0.4,-1.6)-- (m11);
\draw [-,dotted] (m1)--(m11);
\node (rM) at (4,2) {$\mathfrak M$};

\node (rM1)  at (4,1.3){};
\draw[<-,dotted] (rM) -- (rM1);
\end{tikzpicture}
\end{scriptsize}
\end{center}
\caption{\em Directed trees with monitor placement $\chi_{\tt t}$.}
\label{fig:dirtrees}
\end{figure}
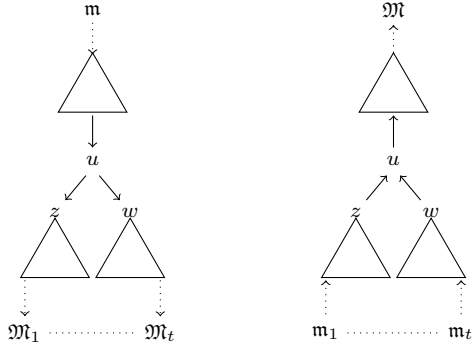

For the upper bound: fix $W=\{u,w\}$ and  $U=\{u\}$. $\pp(U)\subseteq \pp(W)$. Moreover in both cases each path 
passing through $w$ is also touching $u$. Hence  $\pp(\{w\}) \subseteq \pp(\{u\})$. Therefore  $\pp(W) \subseteq \pp(U)$ and then $\pp(U)=\pp(W)$ and $\pp(U)\triangle \pp(W) =\emptyset$.
For the lower bound, let $u$ and $w$ be two distinct nodes in  $\calT_n$.  Let $U=\{u\}$ and $W=\{w\}$. 
Each node in $\calT_n$ is on some path from the root to a leaf.
If $u$ and $w$ lie on different paths, then clearly there are paths in $\pp(U)$ but not in $\pp(W)$. Hence $\pp(U)\triangle \pp(W) \not =\emptyset$. If $u$ and $w$ lie on the same path $p$ and say that $p$ meets $w$ before $u$.
Let $p_w$ be the subpath of $p$ truncated at node $w$. Let $w_1\in N_{\tt o}(w)$ be the neighbour of $w$ lying on $p$. Since $\calT_n$ is LF
 there is necessarily another node $w_2 \not =w_1, w_2\in N_{\tt o}(w)$ and in $\calT_n$ there is a path $q$ from $w_2$ to a leaf. Hence the concatenation of $p_w$ with $q$ is a path from the root to a leaf touching $w$ but not $u$. Hence in $\pp(W)$ but not in $\pp(U)$. Hence 
 $\pp(U)\triangle \pp(W) \not =\emptyset$.
\end{proof}

\noindent{\bf \em Optimality of $\chi_{\tt t}$.}
Notice that the monitor placement $\chi_{\tt t}$ in both cases is optimal. Consider the downward case: if we modify $\chi_{\tt t}$ by removing  one 
output monitor from a leaf, say $u$, then $\mu(\calT_n)=0$: let $v$ be the node parent of $u$ and let $w$ be its other son.
From $\{w\}$ and  $\{v\}$ pass exactly one path. Hence $\mu(\calT_n)<1$.




\subsection{Grids} 
Can we find  topologies whose maximal identifiability is strictly greater than 1?
We analyze $2$-dimensional directed grid $\calH_n$.
\begin{figure}[h!]
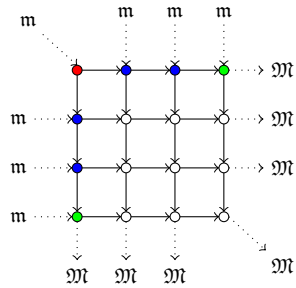

\begin{center}
\begin{scriptsize}
\tikz[scale=1.3] {

\node (m1) at (-0.5,0.5) {$\mathfrak m$};
\node (m2) at (0.5,0.6) {$\mathfrak m$};
\node (m3) at (1,0.6) {$\mathfrak m$};
\node (m4) at (1.5,0.6) {$\mathfrak m$};
\coordinate (a) at (0,0);
\draw[draw=black,dotted,->] (m1)-- (0,0+0.05);
\draw[draw=black,dotted,->] (m2)-- (0.5,0+0.05);
\draw[draw=black,dotted,->] (m3)-- (1,0+0.05);
\draw[draw=black,dotted,->] (m4)-- (1.5,0+0.05);

\node (m13) at (0,-1.5-0.6) {$\mathfrak M$};
\draw[draw=black,dotted,->] (0,-1.5-0.05) -- (m13);
\node (m14) at (0.5,-1.5-0.6) {$\mathfrak M$};
\draw[draw=black,dotted,->] (0.5,-1.5-0.05) -- (m14);
\node (m15) at (1,-1.5-0.6) {$\mathfrak M$};
\draw[draw=black,dotted,->] (1,-1.5-0.05) -- (m15);

\draw [fill=red] (a) circle [radius=0.05];
\coordinate (b) at (0.5,0);
\draw [fill=blue] (b) circle [radius=0.05];
\coordinate (c) at (1,0);
\draw [fill=blue] (c) circle [radius=0.05];
\coordinate (d) at (1.5,0);
\draw [fill=green] (d) circle [radius=0.05];

\draw[draw=black, ->]  (0+0.05,0) -- (0.5-0.05,0);
\draw[draw=black, ->]  (0.5+0.05,0)--(1-0.05,0);
\draw[draw=black, ->]  (1+0.05,0)--(1.5-0.05,0);

\coordinate (a1) at (0,-0.5);
\draw [fill=blue] (a1) circle [radius=0.05];
\coordinate (b1) at (0.5,-0.5);
\draw [] (b1) circle [radius=0.05];
\coordinate (c1) at (1,-0.5);
\draw [] (c1) circle [radius=0.05];
\coordinate (d1) at (1.5,-0.5);
\draw [] (d1) circle [radius=0.05];
\draw[draw=black, ->]  (0+0.05,-0.5) -- (0.5-0.05,-0.5);
\draw[draw=black, ->]  (0.5+0.05,-0.5)--(1-0.05,-0.5);
\draw[draw=black, ->]  (1+0.05,-0.5)--(1.5-0.05,-0.5);
\draw[draw=black, ->]  (0,0-0.05) -- (0,-0.5+0.05);
\draw[draw=black, ->]  (0.5,0-0.05) -- (0.5,-0.5+0.05);
\draw[draw=black, ->]  (1,0-0.05) -- (1,-0.5+0.05);
\draw[draw=black, ->]  (1.5,0-0.05) -- (1.5,-0.5+0.05);

\coordinate (a2) at (0,-1);
\draw [fill=blue] (a2) circle [radius=0.05];
\coordinate (b2) at (0.5,-1);
\draw [] (b2) circle [radius=0.05];
\coordinate (c2) at (1,-1);
\draw [] (c2) circle [radius=0.05];
\coordinate (d2) at (1.5,-1);
\draw [] (d2) circle [radius=0.05];
\draw[draw=black, ->]  (0+0.05,-1) -- (0.5-0.05,-1);
\draw[draw=black, ->]  (0.5+0.05,-1)--(1-0.05,-1);
\draw[draw=black, ->]  (1+0.05,-1)--(1.5-0.05,-1);
\draw[draw=black, ->] (0,-0.5-0.05) -- (0,-1+0.05);
\draw[draw=black, ->]  (0.5,-0.5-0.05) -- (0.5,-1+0.05);
\draw[draw=black, ->]  (1,-0.5-0.05) -- (1,-1+0.05);
\draw[draw=black, ->]  (1.5,-0.5-0.05) -- (1.5,-1+0.05);

\coordinate (a3) at (0,-1.5);
\draw [fill=green] (a3) circle [radius=0.05];
\coordinate (b3) at (0.5,-1.5);
\draw [] (b3) circle [radius=0.05];
\coordinate (c3) at (1,-1.5);
\draw [] (c3) circle [radius=0.05];
\coordinate (d3) at (1.5,-1.5);
\draw [] (d3) circle [radius=0.05];
\draw[draw=black, ->]  (0+0.05,-1.5) -- (0.5-0.05,-1.5);
\draw[draw=black, ->]  (0.5+0.05,-1.5)--(1-0.05,-1.5);
\draw[draw=black, ->]  (1+0.05,-1.5)--(1.5-0.05,-1.5);
\draw[draw=black, ->] (0,-1-0.05) -- (0,-1.5+0.05);
\draw[draw=black, ->]  (0.5,-1-0.05) -- (0.5,-1.5+0.05);
\draw[draw=black, ->]  (1,-1-0.05) -- (1,-1.5+0.05);
\draw[draw=black, ->]  (1.5,-1-0.05) -- (1.5,-1.5+0.05);

\node (m6) at (-0.6,-0.5) {$\mathfrak m$};
\node (m7) at (-0.6,-1) {$\mathfrak m$};
\node (m8) at (-0.6,-1.5) {$\mathfrak m$};
\draw[draw=black,dotted,->] (m6)-- (0-0.05,-0.5);
\draw[draw=black,dotted,->] (m7)-- (0-0.05,-1);
\draw[draw=black,dotted,->] (m8)-- (0-0.05,-1.5);

\node (m9) at (1.5+0.6,0) {$\mathfrak M$};
\draw[draw=black,dotted,->] (1.5+0.05,0) -- (m9);
\node (m10) at (1.5+0.6,-0.5) {$\mathfrak M$};
\draw[draw=black,dotted,->] (1.5+0.05,-0.5) -- (m10);
\node (m11) at (1.5+0.6,-1) {$\mathfrak M$};
\draw[draw=black,dotted,->] (1.5+0.05,-1) -- (m11);
\node (m12) at (1.5+0.6,-2) {$\mathfrak M$};
\draw[draw=black,dotted,->] (1.5+0.05,-1.5) -- (m12);
}
\end{scriptsize}
\end{center}

\caption{\em A directed grid $\calH_4$ with the monitor placement $\chi_{\tt g}$.}
\label{fig:dirgridmon}
\end{figure}
Let us consider the monitor placement $\chi_{\tt g}$ for $\calH_n$ as in Figure \ref{fig:dirgridmon}. 
Formally $\mathfrak m= \{(1,1),\ldots (1,n),(2,1),\ldots,(n,1)\}$ and $\mathfrak M= \{(n,1),(n,2),\ldots (n,n),(1,n), (2,n),\ldots,(n-1,n)\}$.
$(1,1)$ is the only simple source  node and Lemma \ref{lem:bd} can be applied to this case. 

\begin{lemma}
\label{lem:gridub}
Let $n\geq 3$. Then $\mu(\calH_{n}|\chi_{\tt g}) \leq 2$ under $\CAP^-$ and under  $\CSP$.
\end{lemma}
To prove a matching lower bound on  $\mu(\calH_{n}|\chi_{\tt g})$ we prove that any two distinct node sets  $U$ and $W$ of size at most $2$
can be separated by a path in $\pp(\calH_{n}| \chi_{\tt g})$.  
Since in $\CSP$ and $\CAP^-$ we do not allow $\DLP$ paths, we have to
be careful that no path separating $U$ from $W$  can be either $\mathfrak m \!\cdot \! (1,n)(1,n) \! \cdot \! \mathfrak M$ or $\mathfrak m \!\cdot \! (n,1)(n,1) \! \cdot \! \mathfrak M$. With this aim we assign a special role to the complex sources $(1,n)$ and $(n,1)$ (green nodes in Figure \ref{fig:dirgridmon}) and we consider the following assumption which implies (being in fact stronger) that no $\DLP$ paths will separate sets of nodes. 
\begin{assumption}
\label{ass:nodlp}
Nodes $(1,n)$ and $(n,1)$ can be endpoint but never starting point of a path starting in $\mathfrak m$ and ending in $\mathfrak M$. 
\end{assumption}
Let us denote with $V$ the nodes in $\calH_n$ and with $V^-$  the nodes of $\calH_n$  except for $(1,n)$ and $(n,1)$.
By the definition of $\chi_{\tt g}$, $\mathfrak m$ and $\mathfrak M$ are both formed by the border nodes. Hence to fulfill our assumption we define $S=\mathfrak m\setminus\{(1,n),(n,1)\}$, and $T=\mathfrak M$. Given a node  $u$ in $V$, let $S(u)= \{v \in V^- \;| \mbox{$\exists$ a path from $v$ to $u$ in  $\calH_n$} \} $ and 
$T(u)= \{v \in V \;| \mbox{$\exists$ a path from $u$ to $v$ in  $\calH_n$} \}$.

The following Lemmas give a way to build paths  avoiding specific nodes. We always assume 
$n\geq 3$ since otherwise, independently of $d$, $\calH_{n,d}$ would have no node with degree $2d$. 
\begin{lemma}
\label{lem:dgrid-in}
Let  $n\geq 3$. Let $u$ be a node in $V^-$ and $w \in S(u)$ with $w\not =u$. There is a path $p^w_{u}$ from a node in $S$ to $u$ not touching $w$.
 \end{lemma}
 \begin{proof}
 By induction on $S(u)$. If $S(u)=\{u\}$ for some $u \in S$, then $u$ is linked to an input monitor (notice in $S$ we do not have the two mentioned complex sources) and since $u\not =w$, then  $p^w_u$ is the path made by the only node $u$.
In the inductive hypothesis $|N_{\tt i}(u)|=2$, Hence there is $w_1 \in N_{\tt i}(u)$ such that $w_1\not  = w$. Since  $|N_{\tt i}(u)|=2$ , then $S(w_1) \subset S(u)$.  By induction there is a path $p^w_{w_1}$ from $S$ to $w_1$ avoiding $w$. Then define as $p^w_u$, the path concatenating $p^w_{w_1}$ with $u$.\end{proof}
 
A similar proof holds also for the nodes in $T$ reachable from  $u$ without worrying about the two nodes $\{(1,n),(n,1)\}$.
 \begin{lemma}
\label{lem:dgrid-out}
Let  $n\geq 3$. Let $u$ be a node in $\calH_n$ 
and $w \in T(u)$ with $w\not =u$. There is path $q^w_u$ from $u$ to a node in $T$ not touching $w$.
\end{lemma}

Next Claim handles the case when one among $U$ and $W$ contains at least a complex source. In this case to fulfill our assumption, 
we have to show  an $\mathfrak m-\mathfrak M$ path touching exactly one between $U$ and $W$  which is not starting neither with $(1,n)$ nor with $(n,1)$. This immediately implies that this path can be neither  $\mathfrak m \!\cdot \! (1,n)(1,n) \! \cdot \! \mathfrak M$ nor $\mathfrak m \!\cdot \! (n,1)(n,1) \! \cdot \! \mathfrak M$.

\begin{claim}
\label{cl:byc}
Let $U,W$ be  non-empty sets of nodes of $\calH_n$, $n\geq 3$ such that $|U|,|W|\le 2$ and at least one of the complex sources $(1,n)$ or $(n,1)$ belongs to one of them. Then  there is a path from a node in $\mathfrak m$ to a node in $T$ passing though exactly one between $U$ and $W$ fulfilling Assumption \ref{ass:nodlp}. 
\end{claim}

\begin{proof} 
Assume without loss of generality that $(1,n) \in U$ (the case where $(n,1) \in U$ is symmetric). Let $N((1,n))$ be the neighbours of $(1,n)$ in $\calH_n$, i.e. $N((1,n))=\{(1,n-1),(2,n)\}$. We distinguish the following 4 cases:
\begin{enumerate}
\item $W\cap N((1,n))=\emptyset$;
\item $W=N((1,n))$
\item $W  \cap N((1,n)) =\{(1,n-1)\}$
\item $W  \cap N((1,n)) =\{(2,n)\}$
\end{enumerate} 

In each of these cases we find an $\mathfrak m-\mathfrak M$ path touching only one between $U$ and $W$ fulfilling Assumption \ref{ass:nodlp}.

In case (1) and in case (4) the path $(1,n-1)(1,n)$, is a path touching $U$ but not $W$ fulfilling our assumption. 

In case (2) $(1,n)$ is completely surrounded by $W$. So we will build a path touching $W$ but not $U$. If the node $(2,n-1)$ is not in $U$ 
then the path  $(1,n-1)(2,n-1)(2,n)$ proves the claim. If instead $(2,n-1)$ is in $U$ we have to avoid it.
We use here that $n\geq 3$ to build the path starting in $(1,n-1)$ going up to $(1,n-2)$, then going right until the node  $(3,n-2)$ and finally 
going down to the $\mathfrak M$ node $(3,n)$. 

In case (3) we distinguish the following two cases according to whether $U \cap \{(2,n-1),(2,n)\}=\emptyset$ or not.
In the first case the path $(1,n-1)(2,n-1)(2,n)$ touches only $W$ and fulfill the assumption.  In the second case, we follow case (2) and avoid both nodes in  
$\{(2,n-1),(2,n)\}$ using the fact that $n\geq 3$. We start in $(1,n-1)$, go up to $(1,n-2)$, then right up to $(3,n-2)$ and finally down to
$(3,n)$. This path touches $W$ but not $U$ and fulfill the assumption. 
 
\end{proof}

\begin{lemma} (Main Lemma)
\label{lem:dgrid}
Let $n\geq 3$. $\mu(\calH_n | \chi_{\tt g}) \geq 2$.
\end{lemma}
\begin{proof}
Let $V$ be the set of nodes of $\calH_n$. We have to prove that for any $U,W \subseteq V$ with $U\not =W$ and such that $|U|,|W|\leq 2$, $\pp(U) \triangle \pp(W) \not = \emptyset$. It is sufficient to  find a path $p \in \calH_n$ from $\mathfrak m$ to $T$ touching exactly one between $U$ and $W$. 
By Claim \ref{cl:byc} we can assume that neither of $U$ and $W$ contain $(1,n)$ and $(n,1)$. So in the rest of the proof we work only with $S$ and no node will be ever $(1,n)$ and $(n,1)$.
  We split in the following cases: 
\begin{enumerate}
\item at least one between $U$ and $W$ has cardinality $1$;
\item both $U$ and $W$  have cardinality $2$.
\end{enumerate}

\noindent{\bf Case 1}.
Assume wlog that $W=\{w\}$. Since $U \triangle W \not =\emptyset$, then there is a node $u \in U \setminus W$, such that $w\not =u$.
$w$ can be either in (1) $S(u)$,  or (2)  in $T(u)$; or (3)  in $V\setminus(S(u) \cup T(u))$. 
In case (3) any path $p$  from $S$ to $T$  passing through $u$ is not touching $w$ and proves the claim. In case (1) we use Lemma \ref{lem:dgrid-in} to have a path $p^w_u$ from $S$ to $u$ avoiding $W$. Moreover, any path $p$ from $u$ to $T$ is avoiding $w$. Then the composition of $p^w_u$ with $p$ proves the claim. In case (2) any path $p$ from $S$ to $u$ avoids $w$, and Lemma \ref{lem:dgrid-out} guarantees a path $q^w_u$ from $u$ to $T$ avoiding $w$. Hence the composition of the paths  $p$ and  $q^w_u$  proves the claim.

\noindent{\bf Case 2 }. 
Observe that though $U\triangle W \not =\emptyset$, they might  share a node. So there might be two cases: (A) $|U \cap W| =1$ and (B) $|U \cap W| =0$. In case (A) we fix $u$ to be the node of $U$ not in $W$. In case (B) say $U=\{u_0,u_1\}$ we fix $u$ to be the node in $U$ not reachable in $\calH_n$ by the other node in $U$, i.e. the  $u_i$ such that $u_i \not \in S(u_{1-i})$. Notice that this node always exists since the nodes in $U$ cannot reach each other in $\calH_n$.
As in case (1) we divide in three cases according to the position of $W$ wrt $u$.

\begin{enumerate}[i.]
\item $W \subseteq S(u)$;
\item $W \subseteq T(u)$;
\item $|S(u) \cap W |\leq 1$ and  $|T(u) \cap W |\leq 1$;
\end{enumerate}

In case (iii)  a similar argument as above works. Since $|S(u) \cap W |\leq 1$, then either (if $|S(u) \cap W |=0$) any path from $S$ to $u$ avoids $W$, or (if $|S(u) \cap W|=1$) we can apply Lemma \ref{lem:dgrid-in} to find a path $p_u$  from $S$ to $u$ avoiding $W$.  Using $|T(u) \cap W |\leq 1$, a similar argument works for  finding a path $q_u$ from $u$ to $T$ avoiding $W$. Hence the composition of $p_u$ and $q_u$ is a path from $S$ to $T$ passing from $u$ but avoiding $W$.
In case (i)  we further distinguish two cases and fix the $w$ as follows:
\begin{description} 
\item[(A)] $|U \cap W| =1$. $w$ is the only node in $U \cap W$. 
\item[(B)] $w$ is any node in $W$. Denote by $v$ be the other node in $W$.
\end{description}
In case (A) $U=\{u,w\}$ and $W=\{w,v\}$, hence since $u\not =w$, then by Lemma \ref{lem:dgrid-in} there is a path $p_u^w$  from $S$ to $u$ avoiding $w$.  
Moreover, since $W \subseteq S(u)$ any path $q_u$ from $u$ to $T$ avoids $W$. Hence the path $p_u$ which is the concatenation of  $p_u^w$ with  $q_u$ touches $U$ and avoids $W$. This path proves the claim unless $v \in p_u$, and precisely $v \in p^w_u$, since $W \subseteq S(u)$ and $q_u$ lives only in $T(u)$.

If $v\in p_u$ then we modify $p_u$ into a new path $p_v$ touching $v$,  hence the set $W$, but avoiding $U$ and this will prove the claim. To do this we first identify a node $z$ lying on $p^w_u$ before $u$ but  after $v$ and we consider the subpath $p^w_z$  of $p_u^w$ stopping at $z$, hence touching $v$.  The node $z$ is defined as follows: assume $u$ to be the node $u=(x_1,x_2)$ with $x_1,x_2\in [n]$. Since $p_u^{w}$ is ending at $u$ and  $\calH_n$ is directed, there is a first node  $z_1$ in $p_u^{w}$ such that starting from $z_1$ all the nodes 
$z_1,\ldots z_r,u$ of the subpath of $p_u^{w}$ starting at $z_1$ lie either on the same row ($x_1$) or on the same column ($x_2$) of $u$. $z$ is defined to be either $z_1$ or $z_i$ if $v=z_i$ for some $i \in[r]$. The main properties of $z$ are that: $u \in T(z)$ and that $w\not \in p^w_z$. The first is straightforward. For the latter first notice that before $z$ no node on $p^w_z$ can be $w$ because $p^w_z$ is a subpath of  $p^w_u$. Furthermore $w\not \in T(z)$ since $z$ is by definition on the same border of $u$ and hence $S(u) \cap T(z)$ is the set  of nodes $\{z_1,\ldots z_r\}$ and none of them can be $w$.

Since  $u\in T(z)$ and $z \not = u$\footnote{If $v$ is a source node and $v \in N_{\tt i}(u)$, then $z$ is $v$ itself.}, then we can use Lemma \ref{lem:dgrid-out} on $z$ and $u$ to find a path  $q_z^u$ from $z$ to $T$ avoiding $u$. Define $p_v$  the path concatenating $p_z^w$ with $q_z^u$. $p_v$ touches $v$ but avoids both $u$ and $w$,  hence touches $W$ but avoids $U$. Case (A) is proved.

\begin{figure}[h!]
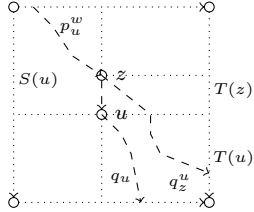

 \begin{center}
\begin{scriptsize}
\tikz[scale=1.3] {
\coordinate (a) at (0,0);
\draw [] (a) circle [radius=0.05];
\coordinate (b) at (2,0);
\draw [] (b) circle [radius=0.05];
\coordinate (c) at (0,-2);
\draw [] (c) circle [radius=0.05];
\coordinate (d) at (2,-2);
\draw [] (d) circle [radius=0.05];
\coordinate (d) at (0.9,-1.1); 
\draw [] (d) circle [radius=0.05];
\coordinate (e) at (0.9,-0.7); 
\draw [] (e) circle [radius=0.05];
\node (u) at (0.9+.2,-1.1) {$u$};
\node (su) at (+0.25,-.75) {\tiny $S(u)$};
\node (tu) at (2.25,-1.55) {\tiny $T(u)$};
\node (tz) at (2.25,-0.85) {\tiny $T(z)$};
\node (pu1) at (0.6,-.2) {\tiny $p^w_u$};
\node (pu2) at (1.1,-1.75) {\tiny $q_u$};
\node (quz) at (1.7,-1.8) {\tiny $q^u_z$};
\node (z) at (0.9+.2,-0.7) {$z$};
\draw[draw=black,dotted, ->]  (0+0.05,0) -- (2-0.05,0);
\draw[draw=black,dotted, ->]  (0+0.05,-2) -- (2-0.05,-2);
\draw[draw=black,dotted, ->]  (0,0-.05) -- (0,-2+.05);
\draw[draw=black,dotted, ->]  (2,0-.05) -- (2,-2+.05);

\draw[draw=black,dotted, -]  (0+0.05,-1.1) -- (0.9-0.05,-1.1);
\draw[draw=black,dotted, -]  (0.9+0.05,-1.1) -- (2-0.05,-1.1);
\draw[draw=black,dotted, -]  (0.9+0.05,-0.7) -- (2-0.05,-0.7);
\draw[draw=black,dashed, ->]  (0.9,-0.7-.05) -- (0.9,-1.1+.05);
\draw[draw=black,dotted, -]  (0.9,-1.1-.05) -- (0.9,-2+.05);

\draw[draw=black,dotted, -]  (0.9,0-.05) -- (0.9,-0.7+.05);
\draw [dashed,->] (0.2,0) -- (0.4,-0.2)-- (.6,-.5)--(0.9,-.7);
\draw [dashed,->] (0.9,-1.1) -- (1.1,-1.3)-- (1.2,-1.5)--(1.3,-2);

\draw [dashed,->] (0.9,-0.7) -- (1.4,-1.1)-- (1.4,-1.3)--(1.5,-1.5)--(2,-1.7);
}
\end{scriptsize}
\end{center}
\caption{\em Case (i).A}
\label{fig:dirgrid}
\end{figure}
In case (B) $w\not \in U$ and let $u_1$ be the other node of $U$. So $U=\{u,u_1\}$ and $W=\{w,v\}$.
The same proof of Case(A) works here too. If $v \in p_u$ however we have to be slightly more careful.
Assume without loss of generality  that $u_1$ appears before $u$ on $p_u$ (the other case is exactly the same swapping $u$ and $u_1$). 
We want to build a path $p_v$ avoiding both $u_1$ and $u$.
Since $v \in S(u)$, we can have two cases: (1)  $v$ is before both $u_1$ and $u$;  and (2) $v$ is in between $u_1$ and $u$. 
Let $t_v$ be the subpath of $p_u$ ending in $v$. In case (1) we use a first time Lemma \ref{lem:dgrid-out} on $v$ and $u_1$ to find a path $t$ from $v$ to $T$ avoiding $u_1$. If $t$ still passes through $u$, then we notice that $u$ cannot be on the border of the grid, 
since otherwise $u_1$ would also be on the same border and hence $t$ would not avoid $u_1$. Hence $u$ is an internal node in the grid.
Let $\{i_1,i_2\}$ be the incoming nodes in $u$. Only one can be $u_1$, say $i_1$. Hence $t$ is entering in $u$ through $i_2$. Let 
$t_2$ be the subpath of $t$ ending at $i_2$. Again by Lemma \ref{lem:dgrid-out} on $i_2$ and $u$ we can find a path $q$ from $i_2$ to $T$ avoiding $u$. The path $p_v$ concatenating $t_2$ with $q$ proves the Claim.
Case (2) is easier. First we apply Lemma \ref{lem:dgrid-in} on $v$ and $u_1$ to have path $t$ from $S$ to $v$ avoiding $u_1$.
Then we apply Lemma \ref{lem:dgrid-out} on $v$ and $u$ to have a path $q$ from $v$ to $T$ avoiding $u$. Then the concatenation
of $t$ with $q$ proves the Claim. In case (ii) a similar argument of case (i) but on $T(u)$ works. 
We left the details to the reader. 
\end{proof}

Together previous Lemma and Lemma \ref{lem:gridub}, imply the following.

\begin{theorem}
\label{thm:dirgrid}
Let $n \in \mathbb N$, $n\geq 3$. Then $\mu(\calH_n |\chi_{\tt g})=2$.
\end{theorem}

Previous result can be easily proved for grids of dimension 
$d>2$ generalizing the definitions and the proofs to the case of a generic $d$ . We left the details to the interested reader.  

\begin{theorem}
\label{thm:hg}
 Let $d\!,n\!\in \!\mathbb N$, $d\!>\!2$ and $n\!\geq\!3$. Then $\mu(\calH_{n,d}|\chi_{\tt g})= d$.
\end{theorem}

\noindent{\bf \em Optimality of $\chi_{\tt g}$.}
In the case $d=2$ we were using $4n-2$ monitors.  
We wonder whether the number of monitors  can be reduced. The answer is essentially no.
Namely, it is easy to see that if in the monitor placement used in $\chi_{\tt g}$ for
Theorem \ref{thm:dirgrid} we remove  the input links to nodes $(1,2)$ and $(2,1)$ (so we have $4n-5$ monitors),
the sets $U=\{(1,2),(2,1)\}$ and $W=\{(1,1)\}$ cannot be separated by any path in $\calH_n$.

In the next section we prove that in the case of undirected grids we can reduce the number of monitors, placing them anywhere and still reaching a high identifiability in terms of the dimension of $\calH_{n,d}$. 
 
\section{Undirected Trees and Grids}
\label{sec:und}

In order to avoid  cases of topologies where identifiability is $0$ we make another assumption on 
the monitor placement for  tree topologies: the tree must be {\em monitor-balanced} in the sense we explain below in Definition \ref{def:mb}.
Notice that next result (Lemma \ref{lem:notmb}) represents in fact a structural limit on the monitor placement of networks whose underlying topology is a tree: if the monitor placement does not make the network monitor-balanced, then the maximal identifiability one can hope for is $0$.

 Let $\calT$ be a tree and $\chi=(\mathfrak m,\mathfrak M)$ be a monitor placement for $\calT$. We say that $\calT$
is an  {\em input tree} (respectively {\em output tree}) with respect to $\chi$  if there is a node of $\calT$ in $\mathfrak m$ (respectively in $\mathfrak M$). 
Notice that a tree can be both an input and an output tree. 

Given a tree $\calT$ and one of its edges $e=(uv)$, let $T^e(u)$ (respectively $T^e(v)$) be the subtrees 
of $\calT$ obtained from cutting in $\calT$  the edge $(uv)$ and taking the tree rooted at $u$ (respectively $v$).

For a node $u \in \calT$, we call the {\em $u$-subtrees of $\calT$} the family of trees $\{T^{(wu)}(w)\}_{w\in N(u)}$.


\begin{definition} (monitor-balanced tree)
\label{def:mb}
A tree $\calT$ is {\em monitor-balanced} under $\chi$  if for each non-leaf node $u$ in $\calT$
the  family $\{T^{(wu)}(w)\}_{w\in N(u)}$ of the $u$-subtrees of $\calT$ contains at least two  input trees and at least two output trees.
\end{definition}

\begin{lemma}
\label{lem:notmb}
If $\calT$ is not monitor-balanced under $\chi$, then $\mu(\calT|\chi) <1$.
\end{lemma}
\begin{proof}
If $\calT$ is not balanced, then there is a non-leaf node $u$ in $\calT$ such that the family $\{T^{(wu)}(w)\}_{w\in N(u)}$ contains
either only one input tree or only one output tree.
There are hence only three possible cases at such a node $u$ that can happen and which are visualized in Figure \ref{fig:NMB}.

\begin{figure*}[!htbp]
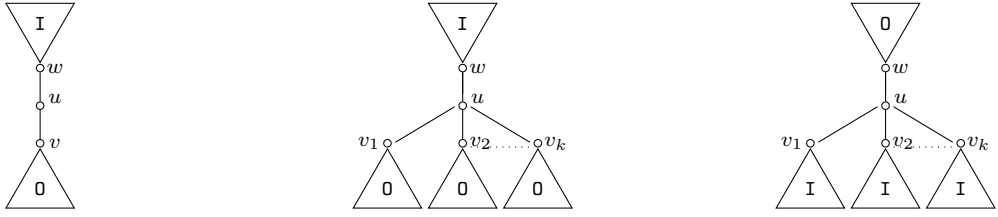

\begin{multicols}{3}
\begin{center}
\begin{scriptsize}
\tikz [triangle/.style = {draw, regular polygon, regular polygon sides=3 },
    border rotated/.style = {shape border rotate=45}] {

\node at  (0.2,0.5) {$w$};
\node (w) at (0,0.55) {};
\coordinate (i1) at (0,0.5);
\draw [] (i1) circle [radius=0.05];
\node[triangle, rotate =180] at (0,1.1) {\tt I};
\node at  (0.2,0.1) {$u$};
\node (u) at (0,-0.05) {};
\draw[-] (w)--(u);
\node (u) at (0,0.05) {};
\coordinate (o1) at (0,0); 
\draw [] (o1) circle [radius=0.05]; 

\node at  (0.2,-0.5) {$v$};
\node (v2) at (0,-0.55) {};
\coordinate (o2) at (0,-0.5); 
\draw [] (o2) circle [radius=0.05]; 
\node[triangle] at (0,-1.1) {\tt O};


\draw[-] (v2)--(u);
}
\end{scriptsize}
\end{center}
\begin{center}
\begin{scriptsize}
\tikz [triangle/.style = {draw, regular polygon, regular polygon sides=3 },
    border rotated/.style = {shape border rotate=45}] {

\node at  (0.2,0.5) {$w$};
\node (w) at (0,0.55) {};
\coordinate (i1) at (0,0.5);
\draw [] (i1) circle [radius=0.05];
\node[triangle, rotate =180] at (0,1.1) {\tt I};
\draw[-] (w)--(u);
\node (u) at (0,-0.05) {};
\draw[-] (w)--(u);
\node at  (0.2,0.1) {$u$};
\node (u) at (0,0.05) {};
\coordinate (o1) at (0,0); 
\draw [] (o1) circle [radius=0.05]; 

\node at  (0.2,-0.5) {$\;v_2$};
\node (v2) at (0,-0.55) {};
\coordinate (o2) at (0,-0.5); 
\draw [] (o2) circle [radius=0.05]; 
\node[triangle] at (0,-1.1) {\tt O};

\node at  (-1.2,-0.5) {$v_1\;$};
\node (v1) at (-1,-0.55) {};
\coordinate (o3) at (-1,-0.5); 
\draw [] (o3) circle [radius=0.05]; 
\node[triangle] at (-1,-1.1) {\tt O};

\node at  (1.2,-0.5) {$\;v_k$};
\node (vk) at (1,-0.55) {};
\coordinate (o4) at (1,-0.5); 
\draw [] (o4) circle [radius=0.05]; 
\node[triangle] at (1,-1.1) {\tt O};
\draw[-,dotted](v2)--(vk);
\draw[-] (v1)--(u);
\draw[-] (v2)--(u);
\draw[-] (vk)--(u);

}
\end{scriptsize}
\end{center}

\begin{center}
\begin{scriptsize}
\tikz [triangle/.style = {draw, regular polygon, regular polygon sides=3 },
    border rotated/.style = {shape border rotate=45}] {

\node at  (0.2,0.5) {$w$};
\node (w) at (0,0.55) {};
\coordinate (i1) at (0,0.5);
\draw [] (i1) circle [radius=0.05];
\node[triangle, rotate =180] at (0,1.1) {\tt O};
\draw[-] (w)--(u);
\node (u) at (0,-0.05) {};
\draw[-] (w)--(u);
\node at  (0.2,0.1) {$u$};
\node (u) at (0,0.05) {};
\coordinate (o1) at (0,0); 
\draw [] (o1) circle [radius=0.05]; 

\node at  (0.2,-0.5) {$\;v_2$};
\node (v2) at (0,-0.55) {};
\coordinate (o2) at (0,-0.5); 
\draw [] (o2) circle [radius=0.05]; 
\node[triangle] at (0,-1.1) {\tt I};

\node at  (-1.2,-0.5) {$v_1\;$};
\node (v1) at (-1,-0.55) {};
\coordinate (o3) at (-1,-0.5); 
\draw [] (o3) circle [radius=0.05]; 
\node[triangle] at (-1,-1.1) {\tt I};

\node at  (1.2,-0.5) {$\;v_k$};
\node (vk) at (1,-0.55) {};
\coordinate (o4) at (1,-0.5); 
\draw [] (o4) circle [radius=0.05]; 
\node[triangle] at (1,-1.1) {\tt I};
\draw[-,dotted](v2)--(vk);
\draw[-] (v1)--(u);
\draw[-] (v2)--(u);
\draw[-] (vk)--(u);

}
\end{scriptsize}
\end{center}
\end{multicols}
\caption{\em The three possible cases of Lemma \ref{lem:notmb} for $u$ when \calT is not monitor-balanced.}
\label{fig:NMB}

\end{figure*}
In all the cases we set $U=\{u\}$ and $W=\{w\}$. Since any path must necessarily go from an input node to an output node, then
$\pp(U)=\pp(W)$. This proves that $\mu(\calT |\chi) <1$.

\end{proof}

On the other hand when $\chi$ is balanced, a similar proof as in Theorem \ref{thm:dirtree} proves that:

\begin{theorem}
\label{thm:undtrees}
Let $\calT$ be a tree and $\chi$ a monitor-balanced monitor placement for $\calT$. Then $\mu(\calT_n| \chi)=1$.
\end{theorem}


\subsection{Grids}


Previous Theorem \ref{thm:hg} on directed grids is true for the undirected case as well.
However, given the higher number of paths  which can be formed after a monitor placement in an undirected grid $\calH_{n,d}$,
it is reasonable to question whether we can reduce the number of monitors but still reaching an identifiability of the order
of the dimension of the grid.

We show that  $2d$ monitors suffice to get maximal identifiability at least $d-1$  and at most $d$ in the case of undirected 
$d$-dimensional grids \emph{for any} monitor placement and under $\CSP$ or $\CAP^-$ routing scheme.

\begin{theorem}
\label{thm:undgrid2}
Let  $n\geq 3$. Then under $\CSP$ and $\CAP^-$ routing schemes,  $d-1\leq \mu(\calH_{n,d}|\chi)\leq d$ 
for any monitor placement $\chi$. 
\end{theorem}
The rest of the section is devoted to the proof  of the theorem for the case $d=2$. The proof 
for $d>2$ is along the same lines and we leave the proof to the reader. 


\begin{proof} (of Theorem \ref{thm:undgrid2})
The upper bound follows from Lemma \ref{lem:ubd}.

For the lower bound we consider the following Claim.
\begin{claim}
\label{claim:inthm}
Let $z_1=(i_1,j_1), z_2=(i_2,j_2), z_3=(i_3,j_3)$ be three nodes in $\calH_{n}$ such that $z_1$ and $z_3$ are distinct nodes. 
There exists a simple path from $z_1$ to $z_3$ touching $z_2$. 
\end{claim}
\begin{proof}
First we consider a rectangle/square of four paths in $\calH_{n}$ such that all these three nodes are lying on the edges of this rectangle/square (see Figure  \ref{fig:claim} for an example). Then we start from the node that we want to be the origin of our path and move along the edge towards our second node that we want to be touched by our path. After reaching the second node we continue moving along the edge which will lead to the third node that our path terminates at. We then build a path from  $z_1$ to $z_3$ touching $z_2$.
\end{proof}

Now we have to prove that independently of what nodes form $\mathfrak m$ and $\mathfrak M$, for any $U,W\subseteq V$ with $U\triangle W\neq \emptyset$ such that $|U|,|W|\leq 1$, then $\pp(U)\triangle \pp(W)\neq \emptyset$. Since $U \not = W$, then there there is at a least  an $\mathfrak m \not \in W$, at least an $\mathfrak M \not \in W$ and of course $u \not \in W$. By Claim \ref{claim:inthm} we get a  simple path from $\mathfrak m$ to $\mathfrak M$ passing through $u$. If this path touches $w$, then we can avoid it. 
If $w$ is an internal node (not on the borders), in order to avoid $w$ we remove this node and all the edges linked to it. Then we have a hole in our grid.  By previous observation, after removing $w$, at least one node in $\mathfrak m$ and one node in $\mathfrak M$ are in the remaining network and they must be different since we do not have degenerate paths allowed neither in $\CSP$ nor in $\CAP^-$. By previous Claim applied to $\mathfrak m$, $\mathfrak M$ and $u$ we have an $\mathfrak m-\mathfrak M$ path in $\calH_{n}$ touching $U$ but not $W$. Notice that if a part of the rectangle/square that we are considering in Claim \ref{claim:inthm} intersects with our hole then we can move along the borders of our hole (see Figure  \ref{fig:gridwithhole} for an example).
If $w$ is on the border but $u$ is an internal node, then by the same argument as above we can touch $w$ and avoid $u$.
If both $w$ and $u$ are on the same border and one of them say $u$ is isolated by $w$, $\mathfrak m$ and $\mathfrak M$ (see Figure  \ref{fig:gridwithhole-corner} for an example), then we remove $u$ and the edges linked to it and again by the same argument as above we have a path $\mathfrak m-\mathfrak M$ touching $W$ but not $U$.
\end{proof}

\begin{figure*}[!htbp]
\begin{multicols}{3}
\begin{center}
\begin{scriptsize}
\tikz[scale=1] {

\coordinate (a) at (0,0);
\draw [] (a) circle [radius=0.05];
\coordinate (d) at (1.5,0);
\draw [] (d) circle [radius=0.05];

\draw[draw=black, -]  (0+0.05,0) -- (0.5,0);
\draw[draw=black, -]  (0.5,0)--(1,0);
\draw[draw=black, -]  (1,0)--(1.5-0.05,0);

\coordinate (b11) at (0.25,-0.25);
\draw [] (b11) circle [radius=0.05];
\coordinate (c11) at (1,-0.5);
\draw [] (c11) circle [radius=0.05];

\draw[draw=black, -]  (0,0-0.05) -- (0,-0.5);
\draw[draw=black, -]  (1.5,0-0.05) -- (1.5,-0.5);
\draw[draw=black, ->]  (0.25,-0.25-0.05) -- (0.25,-1);
\draw[draw=black, ->]  (0.25,-1) -- (0.75-0.05,-1);
\draw[draw=black, ->]  (1,-1) -- (1,-0.5-0.05);
\draw[draw=black, ->]  (0.75+0.05,-1) -- (1,-1);
\draw[draw=black, dotted]  (0.25+0.05,-0.25) -- (1,-0.25);
\draw[draw=black, dotted]  (1,-0.25) -- (1,-0.5+0.05);

\node (z1) at (0.45,-0.35) {$z_1$};
\node (z2) at (0.75,-1.2) {$z_2$};
\node (z3) at (1.2,-0.5) {$z_3$};

\coordinate (b22) at (0.75,-1);
\draw [] (b22) circle [radius=0.05];
\draw[draw=black, -] (0,-0.5) -- (0,-1);
\draw[draw=black, -]  (1.5,-0.5) -- (1.5,-1);

\coordinate (a3) at (0,-1.5);
\draw [] (a3) circle [radius=0.05];
\coordinate (d3) at (1.5,-1.5);
\draw [] (d3) circle [radius=0.05];

\draw[draw=black, -]  (0+0.05,-1.5) -- (0.5,-1.5);
\draw[draw=black, -]  (0.5,-1.5)--(1,-1.5);
\draw[draw=black, -]  (1,-1.5)--(1.5-0.05,-1.5);
\draw[draw=black, -] (0,-1) -- (0,-1.5+0.05);
\draw[draw=black, -]  (1.5,-1) -- (1.5,-1.5+0.05);
}
\end{scriptsize}
\end{center}
\caption{\em Building a path in $\calH_n$ touching three points. \label{fig:claim}}


\begin{center}
\begin{scriptsize}
\tikz[scale=1]{
\coordinate (a) at (0,0);
\draw [] (a) circle [radius=0.05];

\coordinate (d) at (1.5,0);
\draw [] (d) circle [radius=0.05];

\draw[draw=black, -]  (0+0.05,0) -- (0.5,0);
\draw[draw=black, -]  (0.5,0)--(1,0);
\draw[draw=black, -]  (1,0)--(1.5-0.05,0);

\coordinate (b11) at (0.25,-0.25);
\draw [] (b11) circle [radius=0.05];
\coordinate (c11) at (0.75,-0.5);
\draw [] (c11) circle [radius=0.05];

\draw[draw=black, -]  (0,0-0.05) -- (0,-0.5);
\draw[draw=black, -]  (1.5,0-0.05) -- (1.5,-0.5);
\draw[draw=red, -]  (0.25,-0.25-0.05) -- (0.25,-1);
\draw[draw=red, -]  (0.25+0.05,-0.25) -- (0.75,-0.25);
\draw[draw=red, -]  (0.75,-0.25) -- (0.75,-0.5+0.05);
\draw[draw=red, -]  (0.25,-1) -- (0.5,-1);
\draw[draw=black, -]  (0.5,-1) -- (0.75-0.05,-1);
\draw[draw=black, -]  (1,-0.5) -- (1,-1);
\draw[draw=black, -]  (0.75+0.05,-1) -- (1,-1);
\draw[draw=black, -]  (0.5,-0.5) -- (0.75-0.05,-0.5);
\draw[draw=black, dotted]  (0.75,-0.5-0.05) -- (0.75,-1+0.05);
\draw[draw=black, dotted]  (0.5,-0.75) -- (1,-0.75);
\draw[draw=black, -]  (0.75+0.05,-0.5) -- (1,-0.5);
\draw[draw=black, -]  (0.5,-0.5) -- (0.5,-1);
\node (w) at (0.75,-0.75) {$w$};

\coordinate (b22) at (0.75,-1);
\draw [] (b22) circle [radius=0.05];

\draw[draw=black, -] (0,-0.5) -- (0,-1);

\draw[draw=black, -]  (1.5,-0.5) -- (1.5,-1);

\coordinate (a3) at (0,-1.5);
\draw [] (a3) circle [radius=0.05];

\coordinate (d3) at (1.5,-1.5);
\draw [] (d3) circle [radius=0.05];

\draw[draw=black, -]  (0+0.05,-1.5) -- (0.5,-1.5);
\draw[draw=black, -]  (0.5,-1.5)--(1,-1.5);
\draw[draw=black, -]  (1,-1.5)--(1.5-0.05,-1.5);
\draw[draw=black, -] (0,-1) -- (0,-1.5+0.05);
\draw[draw=black, -]  (1.5,-1) -- (1.5,-1.5+0.05);
}
\end{scriptsize}
\end{center} 
\caption{\em Avoiding a hole in $\calH_n.$ \label{fig:gridwithhole}}

\begin{center}
\begin{scriptsize}
\tikz[scale=1] {
\node (u) at (0-.2,0) {$u$};
\node (m) at (0.5,0+.2) {$\mathfrak m$};
\node (w) at (0-.2,-0.5) {$w$};
\node (M) at (0.5+.2,-0.5) {$\mathfrak M$};

\coordinate (a) at (0,0);
\draw [] (a) circle [radius=0.05];
\coordinate (b) at (0.5,0);
\draw [] (b) circle [radius=0.05];
\coordinate (d) at (1.5,0);
\draw [] (d) circle [radius=0.05];
\coordinate (a1) at (0,-0.5);
\draw [] (a1) circle [radius=0.05];
\coordinate (c11) at (0.5,-0.5);
\draw [] (c11) circle [radius=0.05];

\draw[draw=black, -]  (0+0.05,0) -- (0.5-0.05,0);
\draw[draw=red, ->]  (0.5+0.05,0)--(1,0);
\draw[draw=red, ->]  (1,0)--(1,-1);
\draw[draw=red, ->]  (1,-1)--(0,-1);

\draw[draw=black, -]  (1,0)--(1.5-0.05,0);
\draw[draw=black, -]  (1.5,0-0.05) -- (1.5,-0.5);
\draw[draw=black, -]  (0,0-0.05) -- (0,-0.5+0.05);
\draw[draw=black, -]  (0+0.05,-0.5) -- (0.5-0.05,-0.5);
\draw[draw=black, -]  (0.5,0-0.05) -- (0.5,-0.5+0.05);
\draw[draw=red, <-] (0,-0.5-0.05) -- (0,-1);
\draw[draw=black, -]  (1.5,-0.5) -- (1.5,-1);
\coordinate (a3) at (0,-1.5);
\draw [] (a3) circle [radius=0.05];
\coordinate (d3) at (1.5,-1.5);
\draw [] (d3) circle [radius=0.05];

\draw[draw=black, -]  (0+0.05,-1.5) -- (0.5,-1.5);
\draw[draw=black, -]  (0.5,-1.5)--(1,-1.5);
\draw[draw=black, -]  (1,-1.5)--(1.5-0.05,-1.5);
\draw[draw=black, -] (0,-1) -- (0,-1.5+0.05);
\draw[draw=black, -]  (1.5,-1) -- (1.5,-1.5+0.05);
}
\end{scriptsize}
\end{center} 
\caption{\em Avoiding a hole in a corner in $\calH_n.$ \label{fig:gridwithhole-corner}}
\end{multicols}
\end{figure*}

\section{Identifiability through embeddings}
\label{sec:emb}
Let $G=(V,E)$ and $H=(V',E')$ be two DAGs and consider $f$ to be an embedding $G \hookrightarrow_f H$.

Let $\chi$ be a monitor placement in $G$, and $\chi^f$  be the monitor placement for 
$H$ defined by $(f \circ \chi_{\tt i},f \circ \chi_{\tt o})$.
We want to explore what can be said on $\mu(H | {\chi^f})$ in terms of $\mu(G|{\chi})$ 
under the same routing mechanism or even under different routing mechanisms.
 
\begin{figure}[H]
\centering
\begin{center}
\begin{scriptsize}
\tikz[scale=1.4]{
\node (u) at  (-2,-1.5) {$u$};
\node (v) at  (-2,-.5)  {$v$};
\node (uu) at  (-1,-1.5) {$u'$};
\node (vv) at  (-1,-.5)  {$v'$};
\node (zz) at    (-1,-1)   {$z$};
\draw[<-] (u)--(v);
\draw[<-] (uu)--(zz);
\draw[<-] (zz)--(vv);
\draw[->,dotted](u)--(uu);
\draw[->,dotted](v)--(vv);

\node (u) at  (1,-1.5) {$u$};
\node (v) at  (1,-.5)  {$v$};
\node (z) at    (0.5,-1)   {$z$};
\node (uu) at  (2,-1.5) {$u'$};
\node (vv) at  (2,-.5)  {$v'$};
\node (zz) at    (1.5,-1)   {$z'$};
\draw[<-] (uu)--(zz);
\draw[->] (v) -- (z);
\draw[<-] (u) -- (z);
\draw[<-] (zz)--(vv);
\draw[<-] (uu)--(vv);

\draw[->,dotted](u)--(uu);
\draw[->,dotted](v)--(vv);
\draw[->,dotted](z)--(zz);

}
\end{scriptsize}
\end{center}
\caption{{\em Injective and bijective embeddings.}}

\label{fig:monplace}
\end{figure}

As can be seen from the first example in Figure \ref{fig:monplace}, a $1-1$ mapping can map an edge into a  line, so  (on the condition that out-degree of $v$ is $1$ or in-degree of $u$ is $1$) reducing $\mu(H)$ to $0$ disregarding of $\mu(G)$. We then consider $1-1$ and onto mappings for 
the embeddings (also called order-isomorphisms, see \cite{Schr66}). 
We can be tempted to think that under bijective mappings  $G \hookrightarrow H$, 
we can prove that $\mu(G) \geq \mu(H)$. The second example in Figure \ref{fig:monplace} shows that it might  be not always the case:
the sets $\{u',v'\}$ and $\{z'\}$ are separated in $H$ but not their inverse images
$\{u,v\}$ and $\{z\}$ in $G$.

In some cases however we can use embedabbility to say something on identifiability.
In the rest of the section we study what can be said on 
$\mu(H | {\chi^f})$  from $\mu(G | {\chi})$ when $G \hookrightarrow_f H$ and the mapping $f$ is bijective.
To simplify readability we always omit the $\chi$'s, writing simply  $\mu(G)$ and $\mu(H)$.

\smallskip

\noindent{\bf \em  Restricted topologies}.  Consider the following definition 
given in \cite{DBLP:journals/corr/abs-1908-03519,DBLP:journals/corr/abs-1903-10636}.
\begin{definition}(\cite{DBLP:journals/corr/abs-1908-03519,DBLP:journals/corr/abs-1903-10636})
\label{def:rc}
A set of paths \calP is {\em routing consistent} if any two distinct paths $p$ and $p'$ in \calP and any distinct nodes $u$ and $w$ traversed by both paths (if any) 
$p$ and $p'$ follow the same subpath between $u$ and $w$.
\end{definition}
In the directed case we can prove the following result.

\begin{theorem}
\label{thm:emdcr}
Assume that $G=(V,E)$ is a routing consistent directed graph
and that $G \hookrightarrow_f G'$. Then $\mu(G) \leq \mu(G')$.
\end{theorem}
\begin{proof}
Assume $\mu(G') \leq k$. We prove that $\mu(G)\leq k$. 
Since $\mu(G') \leq k$, there are two sets $U',  W' \subseteq V'$ such that $U'\triangle W' \not =\emptyset$, and at least one of them, 
wlog say $U'$, has cardinality  $k+1$, and $\pp_{G'}(U') \triangle \pp_{G'}(W') = \emptyset$.
Fix $U=f^{-1}(U')$ and $W=f^{-1}(W')$. By injectivity of $f$, $U$ has cardinality $k+1$ and $U\triangle W \not =\emptyset$ 
(since otherwise $U'\triangle W'  =\emptyset$). 
Assume by contradiction that $\pp_{G}(U) \triangle \pp_{G}(W) \not =\emptyset$.  That is, there exists a 
path $p$ in $G$ from $S$ to $T$ touching nodes of only one  between $U$ and $W$, say $U$.
Let $p=(u_1u_2)\ldots(u_ru_{r+1})$, $r\geq 1$. Hence $u_i\leq u_{i+1}$ for all $i\in[r]$. Let $u'_i=f(u_i)$.  Clearly if $u_i\in U$, then $u'_i\in U'$. 
Since $f$ is an embedding (i.e. $x\leq y$ iff $f(x)\leq f(y)$), then $u'_i\leq u'_{i+1}$, $u'_1 \in S'$ and $u'_{k+1} \in T'$. 
Hence  there are paths $p'_i$ in $G'$ from $u'_i$ to $u'_{i+1}$ and the path $p'=p'_1,\ldots p'_{k}$ is a path from
$S'$ to $T'$ in $G'$.  If all nodes in $p'$ are in $V'\setminus W$, this is a contradiction with the fact $\pp_{G'}(U') \triangle \pp_{G'}(W') = \emptyset$.
Then there is an $i\in[r]$ such that $p'_i$ is touching a node $w' \in W'$. Hence we have that in $G'$ , $u'_i\leq w' \leq u'_{i+1}$. Since $f$ is an embedding  and since $u_i=f^{-1}(u'_i)$, this means that in $G$, $u_i\leq f^{-1}(w) \leq u_{i+1}$. Then in $G$ there is a path from $u_i$ to $u_{i+1}$ passing through $ f^{-1}(w)$. This contradicts the routing consistency of $G$ since between $u_i$ and $u_{i+1}$ there is another path, the edge that is in $p$.
\end{proof}

\smallskip 

\noindent{\bf \em Restricted embeddings}.
The previous example shows that restricting the class of graphs one can still hope to bound identifiability using embeddability. 
We restrict the class of embeddings, obtaining similar relationships but for broader classes of topologies.
Assume that $f$ is an embedding between two DAGs $G_1=(V_1,E_1)$ and $G_2=(V_2,E_2)$. Let us say that $f$ is 
{\em distance-increasing} (d.i.) if for all $x,y\in V_1$, $d_{G_1}(x,y) \leq d_{G_2}(f(x),f(y))$. Here $d_{G}(x,y)$ is the length of the shortest path between $x$ and $y$ in $G$. We call $f$ {\em distance preserving} (d.p.) if $d_{G_1}(x,y) = d_{G_2}(f(x),f(y))$. 

Distance-increasing of $f$ immediately implies  that the inverse image under $f$ of edges of $G_2$ are edges of $G_1$.

\begin{lemma}
\label{lem:distincr}Let $G=(V,E)$ and $H=(W,F)$.
If $G \hookrightarrow_f H$, $f$ is d.i. and $(w_1,w_2 )\in F$, then $(f^{-1}(w_1),f^{-1}(w_2) )\in E$.
\end{lemma}
 
\begin{theorem}
\label{thm:dremb}
Let $G$ and $G'$ be two DAGs 
such that $G \hookrightarrow_f G'$, where $f$ is a  (d.i.)-embedding. Then $\mu(G)\geq \mu(G')$. 
\end{theorem}
\begin{proof}
Assume $\mu(G) \leq k$, we prove that $\mu(G') \leq k$. Let $S$ and $T$ be respectively the set of source and target nodes in $G$, so that
$\pp_G(u)$ is the set of all paths from $S$ to $T$ in $G$ passing through $u$.
Since $\mu(G) \leq k$, there are two sets $U,  W \subseteq V$ such that $U\triangle W \not =\emptyset$, at least one of them, say $U$, has cardinality  $k+1$, and $\pp_G(U) \triangle \pp_G(W) = \emptyset$.
Fix $U'=f(U)$ and $W'=f(W)$ and let $S'=f(S)$ and $T'=f(T)$. 
By injectivity of $f$, $U'$ has cardinality $k+1$ and clearly $U'\triangle W' \not =\emptyset$ 
(since otherwise $U\triangle W  =\emptyset$). 
Assume by contradiction that $\pp_{G'}(U') \triangle \pp_{G'}(W') \not =\emptyset$.  That means that there exists a 
path $p'$ from $S'$ to $T'$ touching nodes of only one between $U'$ and $W'$, say $U'$. Let $p'=(u'_1,u'_2)\ldots(u'_{r-1},u'_{r})$ and $u_i = f^{-1}(u'_i)$. By Lemma \ref{lem:distincr} for all $i\in [r-1]$, $(u_i,u_{i+1})$ is an edge in $G$ and since
 $f$ is an embedding, then $u_i \in U$, $u_1\in S$ and $u_r \in T$.
But then $p=(u_1,u_2)\ldots(u_ru_{r+1})$ is a path from $S$ to $T$ touching only nodes in $U$.
This is a contradiction with the fact $\pp_G(U) \triangle \pp_G(W) = \emptyset$.
\end{proof}

It is straightforward to see that if $f$ is distance-preserving, then equality holds.
\begin{corollary}
\label{thm:emddr}
Let $G$ and $G'$ be two DAGs 
such that $G \hookrightarrow_f G'$, where $f$ is a  (d.p.)-embedding. Then $\mu(G)=\mu(G')$. 
\end{corollary}

The \emph{dimension} of $G$, $\dim(G)$ is the smallest integer $d$ such that $G \hookrightarrow  \calH_{n,d}$.
 Dushnik and Miller \cite{DM41} proved that for any $n > 1$, the hypergrid $\calH_{n,d}$ has dimension exactly $d$.

We explore how to bound $\mu(G)$ in terms of $\dim(G)$.  Let $G^*$ be the {\em transitive closure}  of a DAG $G$.
\begin{lemma}
\label{lem:tranclo}
Let $G$ and $H$  be  DAGs. If $G$ is closed under transitivity and
$G \hookrightarrow_f H$, then $\mu(G) \geq \mu(H)$.  In particular $\mu(G^*)\geq \mu(G)$.
\end{lemma}
\begin{proof}
Since $G$ is closed under transitivity then the embedding $f$ is necessarily a distance-increasing one. 
Hence the first claim follows by Theorem  \ref{thm:dremb}. The second claim follows since the identity 
is a bijective embedding from $G^*$ to $G$.
\end{proof}

\begin{theorem}
\label{thm:dim}
Let $G$ be a DAG  closed under transitivity. Then $\mu(G) \geq \dim(G)$.
\end{theorem}
\begin{proof}
Let $f$ be the function witnessing the embedding $G \hookrightarrow  \calH_{n,\dim(G)}$.
Since $G$ is closed under transitivity and by Theorem \ref{thm:hg} $\mu(\calH_{n,\dim(G)})=\dim(G)$, the claim follows by previous Lemma.
\end{proof}
\begin{corollary}
\label{cor:emb}
For all graph $G$, for all $k \in \mathbb N$, $\mu(G^k) \geq \mu(G)$.
\end{corollary}

\section{Applications}
\label{sec:algo}

Assume we have to design a network  over  $N\geq 4$ nodes and we aim to have maximal identifiability of failure nodes. Theorem \ref{thm:undgrid2} suggests how to  set edges between the nodes in the network and how to place monitors in such a way to reach an identifiability of at most $\log_3 N$.
Let $n\geq 3$ and set a dimension $d$ in such a way $N=n^d$. Since $n\geq 3$, then $N\geq 3^d$. Hence as long as 
$d\leq \log_3 N$, Theorem \ref{thm:undgrid2} applies. Assume that all values are integers.
Assign an address to each node as a $d$-dimensional vector in $[n]$ and place edges between nodes following  $\calH_{n,d}$.
\subsection{Adding edges to boost node failure identifiability}
\label{subsecalgo:2}
Assume to have a network with very low maximal identifiability of failure nodes (for instance due to a  small minimal in-degree).  We explore the idea to add edges to get better maximal identifiability. 
We propose the following algorithm whose main idea is that of  trying to modify a graph $G$ in order to approach a hypergrid of 
dimension $d$ (a parameter to be tuned), adding edges to the topology  in order to increase the minimal degree to $d$ and choosing $d$ input and $d$ output monitors.  

\begin{tiny}
\begin{algorithm}[ht!]
\caption{\tt AGrid}\label{euclid}
\begin{algorithmic}[1]
\renewcommand{\algorithmicrequire}{\textbf{Input:}}
\renewcommand{\algorithmicensure}{\textbf{Output:}}
\Require $G=(V,E)$, $d$  
\Ensure  $G^{\tt A}=(V,E^{\tt A})$, $I_\mathfrak m \subseteq V$, $I_\mathfrak M \subseteq V$

{ \small /* Boost minimal degree as close to $d$ as possible */ }
\ForAll{$v \in V : \deg(v)<d$}
\label{lin:rani} \State $W=$ choose at random $d-|N(v)|$ node in $V \setminus N(v)$
\ForAll{$w \in W$}
\label{lin:ranf} \State $E=E\cup (v,w)$ \EndFor
 \EndFor

{\small /* Select  input and output nodes*/ }
\label{lin:mon} \For{$i=1\ldots d$} 
\State  Select $x,y\in V$ according to heuristic MDMP
\State $\mathfrak m= \mathfrak m \cup \{x\};  \mathfrak M= \mathfrak M \cup \{y\};$
\State $V = V-\{x,y\}$ 
\EndFor
\end{algorithmic}
\end{algorithm}
\end{tiny}
Given a network $G$, {\tt Agrid}'s aim is to add  a number of random edges so that the minimal degree of the network increases to some suitable  $d=d(N)$ slow-growing function of the number of nodes $N$ in the network. 
{\tt Agrid} assumes to work with a network where monitors are not placed. 
To place monitors we  follow the heuristic of placing monitors  
in the nodes of minimal degree.  We call this heuristic MDMP 
\footnote{The reason of this choice is based on the fact that Theorem \ref{thm:undgrid2} holds for any $\chi$, in particular when monitors are on the corner nodes.}. 
The algorithm receives in input the graph $G$ underlying the network (undirected) and the value $d=d(N)$ and release in output a graph $G^{\tt A}$ whose minimal degree is $d$. The addition of edges is performed between Lines 1 and 4. For each node $v$ with degree smaller than $d$, 
we choose at random a number of neighbours $w$, namely $d-|N(v)|$, and we add an edge in the network between $v$ and $w$, 
keeping updated the set of edges (Line 4).
 To decide what nodes in the network will be monitors, we follow the MDMP strategy. We order the nodes according to  their degree and we choose the first $2d$ to define the sets $\mathfrak m$ and  $\mathfrak M$.  In the {\tt For} in Line 5 we choose $2d$ nodes to be linked to input and output monitors. Notice that a same monitor cannot be chosen to be both in $\mathfrak m$ and in $\mathfrak M$. Using the same heuristic we choose $2d$ monitors in $G$ as well.

\subsubsection{Applicative scenarios for {\tt Agrid}.}
\label{subsec:applied}
We are not aware of other approaches explicitly adding edges to boost identifiability.
Adding a link in some cases may require local or physical access to nodes, access that can be used to check node reliability.
Yet, looking for failing nodes in networks is a process likely to run several times during the 
working-life of  a network, while setting new links to boost identifiability is  an intervention that can be done 
only once, especially if the network is assumed to have a fixed topology.   Hence it makes sense to study cost-benefit tradeoffs for such intervention. We propose below an example of such tradeoffs.

Furthermore, the approaches based on deciding how many monitors to use and where to place them, see for example \cite{DBLP:journals/ton/HeGMLST17}, also might require physical interventions both hardware and software to nodes and links  in the network.  
A difference with our case is that for nodes geographically far apart, the cost of adding a link might be expensive. But this is not always the case.
There are examples of networks where adding links may not require a local intervention (or requires a limited one) and hence our approach  is  reasonable. Its feasibility must be decided again according to some cost-benefit analysis.  
We discuss feasibility of {\tt Agrid} on three possible scenarios: static networks, dynamic networks, subnetworks. 

\smallskip
\noindent{\em Static networks}.
Static networks are grounded on a fixed topology which does not change in the time. In such cases it makes sense to 
analyze the economical feasibility of running  {\tt Agrid}. A way to reduce both the costs and the physical access to nodes and links of the network is an approach suggested for nodes in \cite{DBLP:journals/ton/HeGMLST17}. 

\smallskip
\noindent{\em On-demand link placement.}
Similarly to what is done for nodes in   \cite{DBLP:journals/ton/HeGMLST17}, we can think to employ as additional links in $G^{\tt A}$ {\em temporary links}, which only participate in taking measurements (hence built upon very simple hardware and protocols to transmit simple data packet) and not in other more complex functions. This would reduce the cost of adding links and it might simplify the type of  physical access to the network.  
As a case example to have an idea of the real number of edges to be added, we consider  Table \ref{Tab:enunet}. By adding only 8 edges (over an initial number of 17) one can guarantee  on the real network EuNet an increasing of the maximal identifiability from 0 to 2 (the number of monitors is 3 in both measurements).

\smallskip
\noindent{\em Example of cost-benefit tradeoffs for static networks}.
Assume for the maintenance of a static network $G$ we run end-to-end measurements to detect failure nodes 
for each time $t\in T$.  We dispose of  a function $C_G(e(v,w))$ measuring the cost of adding an edge between nodes $u$ and 
$w$ in $G$. Furthermore we have a function $B_G(t)$ measuring the cost of running a tomography test 
on a network $G$ at time $t \in T$. Let $G^{\tt A}=(V,E^{\tt A})$ be the graph resulting applying {\tt Agrid} to $G$. 
The function $B_G(t)$ may be increasing in the time $t$ but it is meant to be decreasing in the maximal identifiability of $G$.
 We can define a function $\kappa(G,T)$ measuring the tradeoff between the costs and the benefits of applying {\tt Agrid} to $G$ for 
times in $T$ as $$\kappa(G,T) = \frac{\sum_{t \in T}B_G(t)}{\sum_{e\in E^{\tt A}}C_G(e) + \sum_{t \in T}B_{G^{\tt A}}(t)}.
$$ As long as $\kappa(G,T)<1$ an application of {\tt Agrid} produces more benefits than costs on the maintenance of the network $G$.

\smallskip
\noindent{\em Dynamic networks}.
In dynamic networks the  topology is changing in the time according to some rules (but they can be even unpredictable).
They are specified by a sequence of graphs $\{G_t\}_{t \in T}$. For example some cases of wireless networks  are dynamical networks where the underlying topology changes at each given time. 
Nodes are supposed to have a built-in mechanism $\cal M$ to set new links among the  nodes in-sight. 
In such cases we can think to modify  {\tt Agrid} in such a way that links to one node $u$ are added randomly choosing the other nodes among  the nodes reachable from $u$ according to mechanism $\cal M$.
The approach of temporary links would be particularly  suited in the case of dynamic networks, where we can think of adding new edges at each time for each network $G_t$.

\smallskip
\noindent{\em Example of cost-benefit tradeoffs for dynamic networks}.
In this case at each step $t$, we evaluate the 
benefit of adding new links as  $$\beta(t) = B(G^{\tt A}_t) - \sum_{e \in E^{\tt A}} C_{G_t}(e).$$
If $\beta(t)>0$, then adding the edges proposed by {\tt Agrid} would have only benefits. 
When the family $G_t$ is changing according to a specific edge rule, then more refined and global analysis can be 
done on the sequence $\{\beta(t)\}_{t \in T}$.

\smallskip
\noindent{\em Subnetworks}.
We consider  the case when a network is defined as a sub-network of an already given super-network. 
For example a local area subnetwork using the infrastructure of a wider area network.
More formally  a subnetwork  $G'=(V',E')$ is a subgraph of another network $G=(V,E)$ such that $V'\subseteq V$ and 
$E'\subseteq E$.  In such cases we might have efficient, not expensive and not requiring physical access ways of establishing new links 
in $G'$ among those nodes $u,v \in V'$ which are connected in the super-network $G$. When this situation occurs we can think to run {\tt Agrid}
restricting the choice of a new link for the node $u$, randomly among all the other nodes $v$ such that $(u,v) \in E$. In these cases
then the minimal degree in $G$, $\delta(G)$ is an upper bound on the number of maximal links we can add to nodes in $V'$. However in the cases of sub-networks, given that the new link is already present in the super network, we can assume that adding it to the subnetwork will not require any physical access.

\subsection{Applying embeddability results }
Lemma \ref{lem:tranclo} and Corollary  \ref{cor:emb} can be used to understand an upper bound on the maximal identifiability to a given topology $H$ studying the maximal identifiability of  $k$-products of $H$  up to its transitive closure. This might suggest a way to design networks with an improved feature to identify failing nodes.  
Theorem \ref{thm:emdcr} can instead be used in an opposite scenario: if we have a network with strong topology 
restrictions (as the routing consistency) for instance to be renewed,  we might consider of modifying it in a new 
network where to embed the original one with the aim of increasing maximal identifiability. In general keeping the embedding (particularly if onto and 1-1) should not imply any change on the nodes' placements.

\section{Experimental Data from {\tt Agrid}}
\label{sec:perfo}
We study examples of real internet networks whose topologies are publicly available on the data set of the Internet Topology Zoo (\cite{TZ}). 
For any topology $G$ on $N$ nodes and for the parameter $d$ that  we set as either $\log N$ or $\sqrt{\log N}$,
{\tt Agrid} generates  $G^{\tt A}$, the super graph of $G$ that simulates a $d$-hypergrid.
Chosen exactly $2d$ monitors in both $G$ and $G^{\tt A}$ according to heuristic MDMP gives us monitor placements 
 $\chi$ and $\chi^{\tt A}$. We proceed to compute  $\mu(G^{\tt A} | \chi^{\tt A})$ and $\mu(G | \chi)$.
To compute $\mu$ we generate all possible paths, after the monitors are placed. So the number of paths tends to highly grow with the number of nodes. That is the reason why our examples in the Topology Zoo are among those with less than 20 nodes. On networks with more than 25 nodes in the case $d=\log n$  the number of paths in $G^{\tt A}$ quickly reaches the 
number of $5\times 10^6$, making unfeasible our exhaustive search for $\mu$.
The next four subsections collect data according to the following criteria:

\smallskip

\noindent {\bf \em  Real Networks}. In Section \ref{realnet} we compute maximal identifiability, number of paths, minimal degree and number of edges for 
three networks for both the cases $d=\log N$ and  $d=\sqrt{\log N}$ as defined by {\tt Agrid}, so with monitors placed according to MDMP heuristic. 

\smallskip

\noindent {\bf \em Random Graphs}. In Section \ref{randomnet} we run measurements  on random graphs on few nodes (5, 8 and 10). After generating the graphs $G$ and then computing $G^{\tt A}$ with {\tt Agrid}  for 50, 100 or 500 times, we count the fraction of cases where  $\mu$ in the case of $G^{\tt A}$ is increasing or remaining the same (it is never strictly less) and what is the maximal increment of $\mu$ reached in a pair  $(G,G^{\tt A})$. In  these measurements, monitors on $G$ and $G^{\tt A}$ are again computed according to MDMP heuristic. 

\smallskip

\noindent {\bf \em Truncated $\mu$}. In Section \ref{realnet} for each example of network $G$ we compare data only for one pair $(G,G^{\tt A})$. We would like to compare more pairs 
$(G,G^{\tt A})$ for several different random choices of edges forming $G^{\tt A}$ to be more precise about the growth of $\mu(G^{\tt A})$ with respect to $\mu(G)$. However since the running time and the memory are expensive  we cannot launch many runs. So we approximate 
maximal identifiability to a more efficiently computable measure. Loosely speaking we define an $\alpha$-approximation of $\mu(G)$ as that measure  such that  the two sets of paths $U$ and $W$ (compare with  Definition \ref{def:kid}) \emph{must be both} of size at most $\alpha$.  For example if two node sets $U$ of cardinality $1$ and $W$ of cardinality $\alpha+1$ have the same paths passing through, they are instead considered as separable by some paths. This allows us to restrict the search and hence to run our experiments for more pairs $(G,G^{\tt A})$. In  Section \ref{approx} we set $\alpha$ as the \emph{average} degree of the graphs and run 30 different tests for some of the six networks under consideration.

\smallskip

\noindent {\bf \em Random monitors}. In the last section \ref{random} we explore differences between $\mu(G)$ and $\mu(G^{\tt A})$, where monitors are no longer placed according to MDMP heuristic. So we explore the question whether the simulated hypergrid $G^{\tt A}$, for the more significative case of $d=\log N$, has better maximal identifiability also in the case when monitors are not placed according to MDMP.  We explore the effect on $\mu(G)$ and $\mu(G^{\tt A})$ for random placement of monitors on both $G$ and $G^{\tt A}$.  

\subsubsection{Data on real networks}
\label{realnet}
In Tables \ref{Tab:Claranet}, \ref{Tab:enunet} and \ref{Tab:Xchange} we collect data to compare the maximal identifiability of $G$ and $G^{\tt A}$ in 3 of the six networks considered.
$d$ is the dimension of simulated hypergrid and hence the monitors on which we measure maximal identifiability are $2\times 3$, i.e. $|\mathfrak m|=|\mathfrak M|=3$.  In some examples (Table \ref{Tab:Xchange}), when the number of nodes 
is so small that $d\leq \delta(G)$ (so that $G^{\tt A}$ would not change with respect to $G$) we decide to add one dimension to $d$.
The examples  show an increment of the maximal identifiability more evident in the case when $d=\log N$. 
For example in Tables \ref{Tab:enunet} (or Table \ref{Tab:Claranet}) we observe that on a network of 17 edges, adding 6 monitors and 8 links, we pass from not having identifiability at all (using the same number of monitors) to be able to detect uniquely  in $G^{\tt A}$ any two node-failures. This is the consequence of having added in $G$, $8$ new links to raise the minimal degree to $3$. 

\begin{figure*}[ht!]
  \centering
    \scalebox{0.9}{
  \begin{tabularx}{\linewidth}[t]{*{2}X}
    \begin{tabular}[c]{p{\linewidth}}
      \centering
      \begin{tabular}{c|c|c|c|c|}
		\cline{2-5}
 		& \multicolumn{2}{c|}{$d=\sqrt{\log|V|}$} & \multicolumn{2}{c|}{$d=\log|V|$} \\ \cline{2-5} 
 		& $G$ & $G^{\tt A}$ & $G$ & $G^{\tt A}$ \\ \hline
		\multicolumn{1}{|c|}{$\mu$} & \textbf{0} & \textbf{1} & \textbf{1} & \textbf{2} \\ \hline
		\multicolumn{1}{|c|}{$|\calP|$} &  18 & 247  & 39 & 16528 \\ \hline
		\multicolumn{1}{|c|}{$|E|$} &  17 &  22 &  17 & 29 \\ \hline
		\multicolumn{1}{|c|}{$\delta$} & 1  & 2 & 1 & 3 \\ \hline
 		& \multicolumn{2}{c|}{$d,|\mathfrak m|,|\mathfrak M|=2$} & \multicolumn{2}{c|}{$d,|\mathfrak m|,|\mathfrak M|=3$} \\ \cline{2-5} 
	\end{tabular}
    \end{tabular} 
    \quad \quad &
    \quad \quad 
    \begin{tabular}[c]{p{\linewidth}}
      	\centering
      	\begin{tabular}{c|c|c|c|c|}
\cline{2-5}
 & \multicolumn{2}{c|}{$d=\sqrt{\log|V|}$} & \multicolumn{2}{c|}{$d=\log|V|$} \\ \cline{2-5} 
 & $G$ & $G^{\tt A}$ & $G$ & $G^{\tt A}$ \\ \hline
\multicolumn{1}{|c|}{$\mu$} & \textbf{0} & \textbf{1} & \textbf{0} & \textbf{2} \\ \hline
\multicolumn{1}{|c|}{$|\calP|$} & 20 & 40 & 46 & 4917 \\ \hline
\multicolumn{1}{|c|}{$|E|$} & 16 & 17 & 16 & 25 \\ \hline
\multicolumn{1}{|c|}{$\delta$} & 1 & 2 & 1 & 3 \\ \hline
 & \multicolumn{2}{c|}{$d,|\mathfrak m|,|\mathfrak M|=2$} & \multicolumn{2}{c|}{$d,|\mathfrak m|,|\mathfrak M|=3$} \\ \cline{2-5} 
\end{tabular}    
    \end{tabular}     
     \tabularnewline
    \captionof{table}{\em Claranet, $|V|=15$.} 
    \label{Tab:Claranet}&
    \captionof{table}{\em EuNetworks, $|V|=14$.} 
    \label{Tab:enunet} 
\end{tabularx}}
\end{figure*}

\begin{figure*}[ht!]
  \centering
    \scalebox{0.9}{
  \begin{tabularx}{\linewidth}[t]{*{1}X}

 \begin{tabular}[c]{p{\linewidth}}
      \centering
      \begin{tabular}{c|c|c|c|c|}
\cline{2-5}
 & \multicolumn{2}{c|}{$d=\sqrt{\log|V|}$} & \multicolumn{2}{c|}{$d = \log |V|$} \\ \cline{2-5} 
 & $G$ & $G^{\tt A}$ & $G$ & $G^{\tt A}$ \\ \hline
\multicolumn{1}{|c|}{$\mu$} & \textbf{1} & \textbf{1} & \textbf{1} & \textbf{2} \\ \hline
\multicolumn{1}{|c|}{$|\calP|$} & 64 & 108 & 129 & 291 \\ \hline
\multicolumn{1}{|c|}{$|E|$} & 11 & 12 & 11 & 13 \\ \hline
\multicolumn{1}{|c|}{$\delta$} & 1 & 2 & 1 & 3 \\ \hline
 & \multicolumn{2}{c|}{$d,|\mathfrak m|,|\mathfrak M|=2$} & \multicolumn{2}{c|}{$d,|\mathfrak m|,|\mathfrak M|=3$} \\ \cline{2-5} 
\end{tabular}    \end{tabular}
 \tabularnewline
        \captionof{table}{\em DataXchange, $|V|=6$.}   
    \label{Tab:Xchange} 
    \end{tabularx}}
\end{figure*}


\subsubsection{Random graphs}
\label{randomnet}

The Tables \ref{tab:ransqrt} and \ref{tab:ranlog} contain data about maximal identifiability tests on pair $(G,G^{\tt A})$ where $G$ is a random  graph drawn according to Erd\"os-R\`enyi distribution and 
$G^{\tt A}$  is given by the algorithm {\tt Agrid} on input $G$. Monitors are placed according to MDMP heuristic.
On the {\em rows} of the Tables there are the number of generated graphs, while on the {\em columns}, for each of the three cases of 5, 8 and 10 nodes, we separate  cases where $\mu(G^{\tt A})>\mu(G)$, from where $\mu(G^{\tt A})=\mu(G)$. In the square bracket there is the info of the maximal value of $\mu(G^{\tt A})-\mu(G)$
obtained in the tested pairs $(G,G^{\tt A})$.  

Given the small number of nodes, in the case $d=\sqrt{\log N}$ (see Table \ref{tab:ransqrt}) differences between $\mu(G)$ and $\mu(G^{\tt A})$ are more appreciable for smaller values of $n$. This might be due to the fact that the monitors are only $2$  and if they are in different connected components there are no paths between them. Over such few nodes it is much likely that the 2 monitors lie in the same connected component. Instead the situation is completely different when we use $d=\log N$ and the improvement is clearly evident both in the percentage of strictly positive cases that in the maximal increment reached. 
\begin{center}
\begin{figure*}[h!]
\centering
      \scalebox{0.9}{
  \begin{tabularx}{\linewidth}[t]{*{2}X}
    \centering
    \begin{tabular}[c]{ p{\linewidth}}
      \centering
      \begin{tabular}{c|cc|cc|cc}
 & \multicolumn{2}{c|}{5} & \multicolumn{2}{c|}{8} & \multicolumn{2}{c|}{10} \\ \cline{2-7} 
 & \multicolumn{1}{c|}{\textgreater{}} & = & \multicolumn{1}{c|}{\textgreater{}} & = & \multicolumn{1}{c|}{\textgreater{}} & \multicolumn{1}{c|}{=} \\ \hline
50 & {[}2{]}16\% & 84\% & {[}1{]}2\% & 98\% & {[}1{]}2\% & \multicolumn{1}{c|}{98\%} \\ \hline
100 & {[}2{]}16\% & 84\% & {[}1{]}6\% & 94\% & {[}1{]}1\% & \multicolumn{1}{c|}{99\%} \\ \hline
500 & {[}2{]}20\% & 80\% & {[}1{]}4\% & 96\% &  &  \\ \cline{1-5}
\end{tabular}
    \end{tabular}
     \tabularnewline
    \captionof{table}{\em Case $d=\sqrt{\log n}$.} 
    \label{tab:ransqrt}
     \tabularnewline
\end{tabularx}}
\end{figure*}
\end{center}

\begin{figure*}[h!]
\centering
      \scalebox{0.9}{
  \begin{tabularx}{\linewidth}[t]{*{2}X}
    \begin{tabular}[l]{p{\linewidth}}
      \centering
      \begin{tabular}{c|cc|cc|cc}
 & \multicolumn{2}{c|}{5} & \multicolumn{2}{c|}{8} & \multicolumn{2}{c|}{10} \\ \cline{2-7} 
 & \multicolumn{1}{c|}{\textgreater{}} & = & \multicolumn{1}{c|}{\textgreater{}} & = & \multicolumn{1}{c|}{\textgreater{}} & \multicolumn{1}{c|}{=} \\ \hline
50 & {[}2{]}8\% & 92\% & {[}2{]}40\% & 60\% & {[}1{]}16\% & \multicolumn{1}{c|}{84\%} \\ \hline
100 & {[}2{]}18\% & 82\% & {[}2{]}39\% & 61\% & {[}2{]}18\% & \multicolumn{1}{c|}{82\%} \\ \hline
500 & {[}2{]}14\% & 86\% & {[}2{]}34\% & 66\% &  &  \\ \cline{1-5}
\end{tabular}
    \end{tabular} 
    \tabularnewline
    \captionof{table}{\em Case $d=\log n$.} 
    \label{tab:ranlog}
     \tabularnewline
\end{tabularx}}
\end{figure*}

\subsubsection{Truncated maximal identifiability}
\label{approx}
By Definition \ref{def:kid} we know that $\mu(G) \leq k-1$ if there exist two sets $U$ and $W$ with \emph{at least one of them} of size at most $k$  such that $\calP(U)=\calP(W)$. Here we truncate $\mu$ to the measure $\mu_\alpha$ defined in such a way that
$\mu_\alpha(G) \leq \alpha-1$ if there exist two sets $U$ and $W$ \emph{both} of size at most $\alpha$  such that $\calP(U)=\calP(W)$. Searching for sets $U$ and $W$ with $\calP(U)=\calP(W)$ in the case of $\mu_\alpha$ is hence computationally easier than for the case of $\mu$.
We then decide to compare $\mu_\alpha(G)$ with $\mu_\alpha(G^{\tt A})$ fixing $G$ as one of the networks considered  and
generating 30 different $G^{\tt A}$ and eventually test $\mu_\alpha$ on each of these pairs. We fix $\alpha$ to be the \emph{average degree} $\lambda(G)=\lambda$ of  the graph $G$ we compute $\mu_\alpha$ for.

We explain more precisely (see Figure \ref{fig:Matrix}) in what error might produce $\mu_\lambda$ with respect to the real $\mu$ for a graph $G$ over $n$ nodes and why the average  degree is a good choice for $\alpha$.

\begin{figure}[ht!]
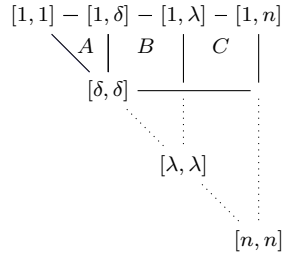

\begin{center}
\begin{scriptsize}
\tikz[scale=1]{
\node (a) at (0,0) {$[1,1]$};
\node(A) at (0.7,-0.4){$A$};
\node(B) at (1.5,-0.4){$B$};
\node(C) at (2.5,-0.4){$C$};
\node (f) at (3,0) {$[1,n]$};
\node (b) at (3,-3) {$[n,n]$};
\node (de) at (1,-1) {$[\delta, \delta]$};
\node (l) at (2,-2) {$[\lambda, \lambda]$};
\node (c) at (1,0) {$[1,\delta]$};
\node (d) at (2,0) {$[1,\lambda]$};
\node (g)  at (3,-1) {};
\node (e) at (2,-1) {}; 

\draw[fill=blue] (a) -- (c) -- (de) -- (a) --  cycle;

\draw[draw=black, -,dotted] (de) -- (l);
\draw[draw=black, -,dotted] (g) -- (b);
\draw[draw=black, -,dotted] (l) -- (b);

\draw[draw=black, -] (a) -- (c);
\draw[draw=black, -] (c) -- (d);
\draw[draw=black, -] (d) -- (f);
\draw[draw=black, -] (d) -- (e);
\draw[draw=black, -,dotted] (e) -- (l);
\draw[draw=black, -] (de) -- (g);
\draw[draw=black, -] (c) -- (de);
\draw[draw=black, -] (f) -- (g);
}
\end{scriptsize}
\end{center}
\caption{\em The matrix $M$. 
}
\label{fig:Matrix}
\end{figure}

Consider an $n\times n$ matrix $M$  containing the following data: the entry $(i,j)$ is  the set of all possible pairs $(U,W)$ with 
$U,W \subseteq V$, $U \not =W$ such that  $|U|=i$ and $|W|=j$. $M$ is symmetric and we are interested only in one of its half, say the upper part that we call $M$ (since the whole matrix will be never in use). For two integers $\delta$ and $\lambda$, $1\leq \delta \leq \lambda \leq n$ let: $M[\delta]$ be the sub-matrix of $M$ between rows $1$ and $\delta$ ( Zones $A$,$B$ and $C$  in Figure \ref{fig:Matrix}); $M[\delta,\lambda]$ be the sub-matrix of $M$ between row $\delta$ and column $\lambda$ (Zones $A$ and $B$  in Figure \ref{fig:Matrix});  
$M[\delta,\lambda^+]$  be the sub-matrix of $M$ where  rows are up to  $\delta$ and columns greater than $\lambda$ (Zone $C$ in Figure \ref{fig:Matrix}).
Since $\mu(G) \leq \delta(G)=\delta +1$, then there exists a pair $(i,j)$ in $M[\delta]$ (i.e. in Zones $A$,$B$ and $C$ of $M$) such that $\pp(U) = \pp(W)$. 
Searching only for $\mu_\lambda$ corresponds to the restricted search over only Zones $A$ and $B$ of $M$, i.e. $M[\delta,\lambda]$. 
Hence the potential error of this search is only in the pairs in $M[\delta,\lambda^+]$ (i.e.  Zone $C$). 
This quantity decreases while $\lambda-\delta$ is growing. More precisely, in each entry $M[i,j]$ of $M$ we are storing $\zeta(i,j)={n \choose i}\left[{n \choose j}-1\right]$ pairs. Hence in $M[\delta,\lambda^+]$ we have $\sum_{i=1}^{\delta}\sum_{j=\lambda+1}^{n} \zeta(i,j)$ pairs, while the search space for the real $\mu$ (Zones $A$,$B$ and $C$) is of cardinality 
$\sum_{i=1}^{\delta}\sum_{j=i}^{\delta} \zeta(i,j) + \sum_{i=1}^{\delta}\sum_{j=\delta}^{n} \zeta(i,j)$. Hence the maximal fraction of pairs on which $\mu_\lambda$ can fail with respect to $\mu$
is:
$$
\frac{\sum_{i=1}^{\delta}\sum_{j=\lambda+1}^{n} \zeta(i,j)}{\sum_{i=1}^{\delta}\sum_{j=i}^{\delta} \zeta(i,j) + \sum_{i=1}^{\delta}\sum_{j=\delta}^{n} \zeta(i,j)}. 
$$
In the Tables \ref{tab:claranet-501}, \ref{tab:grid-501} and \ref{tab:eunet-501} we collect data on 30 tests showing the percentage of cases of where $\mu_{\lambda(G)}$ is equal to the values indicated in the columns.   For these Tables we consider only the case where $d=\log N$. In the square bracket (e.g. $[3 ]G$) we indicate the average degree 
of the graph. 


\begin{figure*}[ht!]
\centering
     \scalebox{0.9}{
  \begin{tabularx}{\linewidth}[t]{*{3}X}
    \begin{tabular}[c]{p{\linewidth}}
      \centering
     \begin{tabular}{c|ccc|}
\cline{2-4}
$G$ \textbackslash$\mu_\lambda$ & \multicolumn{1}{c|}{\textbf{0}} & \multicolumn{1}{c|}{\textbf{1}} & \textbf{2} \\ \hline
\multicolumn{1}{|c|}{[2]$G$} & 100\% & 0\% & 0\% \\ \hline
\multicolumn{1}{|c|}{[4]$G^{\tt A}$} & 0\% & 65\% & 55\% \\ \hline
\end{tabular}
    \end{tabular} 
    &  
    \centering
    \begin{tabular}[c]{p{\linewidth}}
      \centering
     \begin{tabular}{c|ccc|}
\cline{2-4}
$G$ \textbackslash$\mu_\lambda$ & \multicolumn{1}{c|}{\textbf{0}} & \multicolumn{1}{c|}{\textbf{1}} & \textbf{2} \\ \hline
\multicolumn{1}{|c|}{[4]$G$} & 0\% & 13\% & 86\% \\ \hline
\multicolumn{1}{|c|}{[4]$G^{\tt A}$} & 0\% & 6\% & 93\% \\ \hline
\end{tabular}   
 \end{tabular} &
     \centering
    \begin{tabular}[c]{p{\linewidth}}
      \centering
      \begin{tabular}{c|ccc|}
\cline{2-4}
$G$ \textbackslash$\mu_\lambda$ & \multicolumn{1}{c|}{\textbf{0}} & \multicolumn{1}{c|}{\textbf{1}} & \textbf{2} \\ \hline
\multicolumn{1}{|c|}{{[}2{]}$G$} & 100\% & 0\% & 0\% \\ \hline
\multicolumn{1}{|c|}{[3]$G^{\tt A}$} & 0\% & 100\% & 0\% \\ \hline
\end{tabular}
    \end{tabular} 
    \tabularnewline
    \captionof{table}{\em Claranet, $|V| =15$.}     
    \label{tab:claranet-501} &
    \captionof{table}{\em GridNetwork, $|V| = 7$.} 
    \label{tab:grid-501}&
    \captionof{table}{\em EuNetwork, $|V| = 7$.} 
    \label{tab:eunet-501}
     \tabularnewline
\end{tabularx}}
\end{figure*}

\subsubsection{Random monitors}
\label{random}
MDMP is a simple heuristic for monitor placement. However  the lower bound of Theorem \ref{thm:undgrid2} holds for any monitor placement.
We try to give some evidence that {\tt Agrid} is a good strategy for boosting maximal identifiability independently of where monitors are 
placed. In the following Tables we collect results for percentage of values of $\mu(G)$ on 20 random placements of monitors both in $G$ and $G^{\tt A}$.  Tables \ref{tab:claranet-5}, \ref{tab:epoch-5-3} and  \ref{tab:getnet-5}  show that moving to $G^{\tt A}$ gives an improvement in the maximal identifiability, independently of the monitor 
placement. Also in this case the data  are computed only for the most significative case of $\log N$.
 \begin{figure*}[ht!]
 \centering
 \scalebox{0.9}{
  \begin{tabularx}{\linewidth}[t]{*{3}X}
    \begin{tabular}[c]{p{\linewidth}}
      	\centering
     	\begin{tabular}{c|ccc|}
   		\cline{2-4}
    		$G$ \textbackslash$\mu$ & \multicolumn{1}{c|}{\textbf{0}} & \multicolumn{1}{c|}{\textbf{1}} & \textbf{2} \\ \hline
    		\multicolumn{1}{|c|}{$G$} & 20\% & 80\% & 0\% \\ \hline
    		\multicolumn{1}{|c|}{$G^{\tt A}$} & 0\% & 0\% & 100\% \\ \hline
    	\end{tabular}
      \end{tabular} &    
     \centering
    \begin{tabular}[c]{p{\linewidth}}
      	\centering
      	\begin{tabular}{c|ccc|}
		\cline{2-4}
		$G$ \textbackslash$\mu$ & \multicolumn{1}{c|}{\textbf{0}} & \multicolumn{1}{c|}{\textbf{1}} & \textbf{2} \\ \hline
		\multicolumn{1}{|c|}{$G$} & 100\% & 0\% & 0\% \\ \hline
		\multicolumn{1}{|c|}{$G^{\tt A}$} & 0\% & 80\% & 20\% \\ \hline
		\end{tabular}    
	\end{tabular} &
	\centering	
\begin{tabular}[c]{p{\linewidth}}
      	\centering
      	\begin{tabular}{c|ccc|}
		\cline{2-4}
		$G$ \textbackslash$\mu$ & \multicolumn{1}{c|}{\textbf{0}} & \multicolumn{1}{c|}{\textbf{1}} & \textbf{2} \\ \hline
		\multicolumn{1}{|c|}{$G$} & 0\% & 100\% & 0\% \\ \hline
		\multicolumn{1}{|c|}{$G^{\tt A}$} & 0\% & 10\% & 90\% \\ \hline
	\end{tabular}
    \end{tabular}	
    \tabularnewline
    \captionof{table}{\small \em Claranet, $|V| = 15$, $\mathfrak m,\mathfrak M,d=3$.} 
    \label{tab:claranet-5} &
    \captionof{table}{\small \em EuNetworks, $|V| = 14$, $\mathfrak m,\mathfrak M,d=3$.} 
    \label{tab:epoch-5-3} &
     \captionof{table}{\small \em GetNet, $|V| = 9$, $\mathfrak m,\mathfrak M,d=3$.} 
    \label{tab:getnet-5}
  \end{tabularx}}
\end{figure*}

\section{Discussion}
 \label{sec:final}
\noindent{\bf \em  DLP paths}.
First notice that if we need to consider the single node $v$ as a potential path, then this should be viewed as
a \emph{degenerate loop-path} made by the loop edge $(vv)$. This is because an end-to-end measurement path starting and ending in the same monitor defines a loop.    Inclusion of $\DLP$s among the allowed paths may have important effects on the maximal identifiability:  If $v$ is a $\DLP$ node, then the set  $\{v\}$ would have a maximal \emph{local} identifiability, as high as the total number of nodes in the set of paths. 
This is  because the "path" made by $v$ alone distinguishes the set $\{v\}$ from any other set of nodes different from $\{v\}$.  
Hence $\DLP$s allow to trivially reach high local-identifiability on nodes which are $\DLP$s. $\DLP$ were implicitly used in several previous works 
to raise the maximal identifiability by placing monitors appropriately.  However we think that from the point of view of capturing the combinatoric of maximal identifiability of graph topologies $\DLP$ are essentially not influent. We list some reasons that justify our decision of not considering them in our analysis:  (1) if all nodes in $G$ were $\DLP$, we would trivially  solve the problem of identifying node failures. We call this a {\em $\DLP$-strategy} for monitor placement ; (2) Notice that: (1)$\DLP$ nodes can be immediately verified for a failure (and corresponding equations immediately solved) independently of the rest of the topology. 
(3) $\DLP$ nodes allow to distinguish sets of paths $\pp(U)$ and $\pp(W) $ for $U$ and $W$ \emph{differing on $v$}. 
But to infer node failure {\em globally} in a topology, $\DLP$ nodes are as relevant as any other node; 
(3) To use a $\DLP$-strategy, any node of $G$ must be linked to two monitors. This is highly not efficient and it makes the identifiability problem trivial and unrelated from the set of paths; 
 (4) The asymmetry to \emph{force} a single node as a path creates side-effects in studying identifiability with respect to a  set of measurement paths corresponding to \emph{real paths} in a graph $G$. 

\smallskip

\noindent {\bf \em Routing mechanisms}.
As discussed in \cite{DBLP:conf/infocom/RenD16} controllable routing is a main issue in end-to-end path 
measurement. In \cite{DBLP:journals/ton/Hu0W0L0ZG16} it was recently introduced XPath a practical way to implement explicit path control. Details can be found on both \cite{DBLP:conf/infocom/RenD16} and \cite{DBLP:journals/ton/Hu0W0L0ZG16}.  In our case we briefly mention that  XPath  can easily  implement the routing 
$\CAP^-$ (as well as $\CSP$). Since it explicitly identifies an end-to-end path with a path ID and preinstalls all the desired path IDs between any source-destination pairs in order to recognize if the signal is received through a valid path, hence it is sufficient to disallow DLP paths in the ID table stored by XPath in receiving nodes to implement $CAP^-$.  
Of course XPath is designed specifically to handle efficiently small sets of paths in network with a huge number of 
paths. Nevertheless it can be perfectly applied to our case; for example guaranteeing that the paths $\pp(G^{\tt A}| \chi^{\tt A})$ in the network output of  {\tt Agrid} are under the specified routing mechanism.  

\smallskip

\noindent{\bf \em Further research}. We shortly address some directions related to our approach which might be further explored in the analysis of identifiability of failure nodes. 
 In 1982  \cite{DBLP:journals/jal/YanezM99} showed that for $k\geq3$ to test if a partial order has dimension $\leq k$ is $NP$-complete. 
 Nevertheless there are some algorithms  to compute the dimension of poset \cite{doi:10.1137/0603036,DBLP:journals/jct/TrotterM77} which are practically used. It would be interesting to further explore  connections between boolean network tomography  and poset  dimension theory to get better estimates on the maximal identifiability for DAG network topologies.  
 It is a well-known result \cite{Schnyder:1990:EPG:320176.320191} that planar graphs over $n$ nodes  can be embedded  through a straight line embedding into a $(n-2) \times (n-2)$ 2-dimensional grid. It seems not difficult to see that our results on embeddability  can be generalized to obtain a lower bound of 2 for the maximal identifiability when a network is a planar 
 graph. Connections with dimension might also be explored in the case of planar networks \cite{DBLP:journals/jgt/FelsnerT05}.
A $k$-Transitive-Closure-Spanner of a graph $G$ is a graph $H$ with a small diameter - $k$-  that preserves the connectivity of the original graph. 
The edges of the transitive closure of $G$, added to $G$ to obtain a TC-spanner, are called shortcuts and the parameter $k$ is called the {\em stretch}.
These graphs and their relations with dimension of poset were recently studied in \cite{DBLP:journals/combinatorica/BermanBGRWY14}. From our results it is clear that adding edges to a graph $G$ can strength the potential of failure identifiability. Are $k$-TC-Spanners and in particular Steiner-$k$-TC-Spanners (see \cite{DBLP:journals/combinatorica/BermanBGRWY14}) useful to maximize failure identifiability of a network?

 {\tt Agrid} might  be explored using different heuristics: for example (1) adding  edges to a node $v$ only with nodes of degree at most $d-1$; (2)  adding edges between nodes $u$ and $v$ only if their shortest distance is greater than a given value; (3) adding edges only if a  planarity condition is respected. An interesting question, relevant to apply XPath with {\tt Agrid},  is how to efficiently determine the minimum number of measurement paths sufficient to identify all the failures after {\tt Agrid} is applied. Finally, new connections between maximal node identifiability and vertex connectivity were recently found in \cite{GRZ19} which can be further explored in connection with embeddability.

\section*{Acknowledgments}
We  thank Liang Ma and Michele Zito  
for discussions about this work and Viviana Arrigoni, Annalisa Massini and Michele Zito for reading different versions of this work and
sending us their comments which contribute to improve the paper.

\bibliographystyle{abbrv}
\bibliography{biblioBNT}

\end{document}